\documentclass[12pt]{article}
\usepackage{amsmath}
\usepackage{graphicx}
\usepackage{enumitem}
\usepackage{natbib}
\usepackage{overpic}
\usepackage{url} 
\usepackage{multirow}
\usepackage{amsmath}
\usepackage{commath}
\usepackage{amsthm}
\usepackage{amssymb}
\usepackage{setspace,mathrsfs}
\usepackage{algorithm}
\usepackage{algorithmic}
\usepackage{url,amsmath,amsfonts,amssymb}
\usepackage{caption,algorithm,graphicx}
\usepackage{listings,color,relsize,multirow,rotating,algorithm}
\usepackage{booktabs}
\usepackage{bbm}
\usepackage{subcaption}

\newcommand{\blind}{1}

\definecolor{purple}{RGB}{250,000,180}

\RequirePackage[colorlinks,
            linkcolor=blue,
            anchorcolor=blue,
            citecolor=blue
            ]{hyperref}

\def\ttop{^{\top}}

\usepackage{mathrsfs}

\def \R {\mathbb{R}}
\def \P {\mathbf{P}}
\def \E {\mathbf{E}}

\newcommand{\ts}{\textstyle}

\newcommand{\var}{\textup{var}}
\newcommand{\Var}{\textup{var}}
\newcommand{\Cov}{\textup{cov}}
\newcommand{\dK}{d_{\textup{K}}}
\renewcommand{\L}{\mathcal{L}}

\newtheorem{lemma}{Lemma}

\newtheorem{proposition}{Proposition}
\newtheorem{assume}{Assumption}

\newtheorem{theorem}{Theorem}

\begin{document}

\def\spacingset#1{\renewcommand{\baselinestretch}%
{#1}\small\normalsize} \spacingset{1.5}


\if1\blind
{
  \title{\Large\bf Robust Max Statistics for High-Dimensional Inference}
  \author{Mingshuo Liu and Miles E. Lopes\\
    University of California, Davis}
\date{}
  \maketitle
} \fi

\if0\blind
{
  \bigskip
  \bigskip
  \bigskip
  \begin{center}
      {\LARGE\bf Robust Max Statistics\\[0.2cm] for High-Dimensional Inference}
\end{center}
  \medskip
} \fi

\begin{abstract}
\singlespacing
Although much progress has been made in the theory and application of bootstrap approximations for max statistics in high dimensions, the literature has largely been restricted to cases involving light-tailed data. To address this issue, we propose an approach to inference based on \emph{robust max statistics}, and we show that their distributions can be accurately approximated via bootstrapping when the data are both high-dimensional and heavy-tailed. In particular, the data are assumed to satisfy an extended version of the well-established $L^{4}$-$L^2$ moment equivalence condition, as well as a weak variance decay condition. In this setting, we show that \emph{near-parametric} rates of bootstrap approximation can be achieved in the Kolmogorov metric, \emph{independently of the data dimension}. Moreover, this theoretical result is complemented by encouraging empirical results involving both Euclidean and functional data.

\end{abstract}

\noindent%



\spacingset{1} 
\section{Introduction}\label{sec:intro}
Over the past decade, distributional approximation results for max statistics have become a prominent topic in high-dimensional inference. A prototypical example of such a statistic has the form $\max_{1\leq j\leq p}\sqrt n(\bar X_j-\mu_j)$, where $\bar X\in\R^p$ is the sample mean vector of $n$ observations and $\mu=\E(\bar X)$, but numerous variants arise in diverse contexts. Indeed, one of the main drivers of research on this topic is that many high-dimensional inference tasks can be unified within the problem of approximating the distribution of $\max_{1\leq j\leq p}\sqrt n(\bar X_j-\mu_j)$, or some adaptation of it. For instance, such approximations can be directly applied to construct simultaneous tests and confidence intervals for  coordinate-wise means $\mu_1,\dots,\mu_p$. More broadly, other applications include detection of treatment effects~\citep{Sun:2022}, error estimation for sample covariance matrices~\citep{Lopes:Bernoulli}, post-selection inference~\citep{Kuchibhotla:2020:PSI}, change-point detection~\citep{Yu:2021}, confidence bands in non-parametric regression~\citep{Singh:2023}, tests for shape restrictions~\citep{Chetverikov:shape}, and more. Meanwhile, another major reason why max statistics have attracted growing interest is that
bootstrap methods can accurately approximate their distributions when $p$ is much larger than the sample size $n$, which has been demonstrated by a cascade of theoretical advances~\cite{CCK:2013, CCK:2017,Deng:2020,Kuchibhotla2020high,lopes2020bootstrapping,Kuchibhota_2021,Lopes:2022:AOS,CCK:2023,Fang:2023,Koike:2024}.

 Despite the substantial innovations that have been made in bootstrap approximations for max statistics, there is an Achilles heel that continues to hinder much of the research in this area. Namely, there is a widespread reliance on the assumption that the covariates have light tails, e.g.,~sub-Gaussian or sub-exponential. 
Moreover, there are empirical and theoretical results suggesting that light tails \emph{are necessary} for bootstrap methods to successfully approximate the distributions of conventional max statistics in high dimensions~\citep{zhang2017gaussian,Giessing:2020,Preinerstorfer:2024}.
For instance, the simulations in~\citep{Giessing:2020} show that the Gaussian multiplier bootstrap performs poorly for $\max_{1\leq j\leq p}\sqrt n|\bar X_j-\mu_j|$ when the covariates have heavy tails and $p\gg n$. From a theoretical standpoint, it has also been proven that there is a moment-dependent phase transition governing the success of Gaussian approximations for $\max_{1\leq j\leq p}\sqrt n(\bar X_j-\mu_j)$~\citep{zhang2017gaussian,Preinerstorfer:2024}. That is, if $W\in\R^p$ is a centered Gaussian vector having the same covariance matrix as $\sqrt n(\bar X-\mu)$, then the Kolmogorov distance between $\max_{1\leq j\leq p}W_j$ and $\max_{1\leq j\leq p}\sqrt n(\bar X_j-\mu_j)$  may or may not vanish in the limit that $n$ and $p$ jointly diverge, depending on whether the covariates have enough moments. This breakdown of Gaussian approximations suggests that similar behavior should occur for bootstrap approximations---especially in the case of the Gaussian multiplier bootstrap, which seeks to mimic the behavior of $\sqrt n(\bar X-\mu)$ by generating random vectors from a centered Gaussian distribution whose covariance matrix is an estimate for that of $\sqrt n(\bar X-\mu)$.

Due to the issues just mentioned, there are strong motivations to extend bootstrap methods involving max statistics so that they can be applied reliably to high-dimensional data with heavy tails.
However, the research in this direction is still at a very early stage, and there are just a couple of previous works that have given it attention. The first of these works briefly outlined an approach that combines truncation with permutation-based sampling~\citep{Lou:2017}, but it was ultimately not pursued as a practical method for inference. More recently, the state-of-the-art paper~\citep{Fan:2023} proposed a weighted bootstrap for a max statistic of the form $\max_{1\leq j\leq p}\sqrt n|\hat\theta_j-\theta_j|$, where $\theta_j$ denotes the so-called ``pseudomedian'' of the $j$th covariate, and $\hat \theta_j$ is the classical Hodges-Lehmann estimator for $\theta_j$~\citep{Hodges:1963}.

While the approach in~\citep{Fan:2023} achieved major progress by delivering robust simultaneous inference for $\theta_1,\dots,\theta_p$, it still has some essential limitations. One is that the  pseudomedians can be unsatisfactory substitutes for the means $\mu_1,\dots,\mu_p$, particularly in cases of asymmetric distributions, for which $\theta_j$ and $\mu_j$ may be quite different. A related issue is that an approach based on pseudomedians does not extend naturally to suprema of zero-mean empirical processes, which appear frequently in applications of bootstrap approximations for max statistics~\citep{CCK_bands,Chen:ridges,Han:Bernoulli,Chen:Jackknife,Dette:Banach,Lopes:Bernoulli,Giessing:processes}. Another issue is that the method in~\citep{Fan:2023} produces simultaneous confidence intervals for $\theta_1,\dots,\theta_p$ that are only theoretically justified when they all have the same width, which is impractical if the covariates fluctuate over different scales. Lastly, the available theoretical analysis for $\max_{1\leq j\leq p}\sqrt n|\hat\theta_j-\theta_j|$ establishes a near $n^{-1/4}$ rate of bootstrap approximation in the Kolmogorov metric, which does not align with other recent results for max statistics that establish near $n^{-1/2}$ rates in the setting of light-tailed data~\cite{lopes2020bootstrapping,Lopes:2022:AOS,CCK:2023,Fang:2023,Koike:2024}.

In the current paper, we propose to bootstrap a robust max statistic that enables simultaneous inference on the means $\mu_1,\dots,\mu_p$ and overcomes the difficulties described above. Our approach is designed in terms of three ingredients: truncation, partial standardization, and the median-of-means (\textsc{mom}) technique~ \citep{Nemirovsky1983,lugosi2019mean}. To briefly lay out the main ideas, let $X_1,\dots,X_{n}\in\R^p$ be i.i.d.~observations with $\var(X_{1j})=\sigma_j^2$, and  $\mu_j=\E(X_{1j})$ as before. Also, let $\hat\sigma_1^2,\dots,\hat\sigma_p^2$ denote variance estimates that will be constructed from a small hold-out set via \textsc{mom}, and  define the truncation function $\varphi_t(x)=\textup{sgn}(x)(|x|\wedge t)$ for any $x\in\R$ and $t\geq 0$, where $a\wedge b=\min\{a,b\}$. In this notation, the proposed robust max statistic is defined by
\begin{equation}\label{eqn:maxstatdef}
\small
\mathcal{M}_n=\max_{1\leq j\leq p} \sum_{i=1}^n 
 \frac{\varphi_{\hat t_j}(X_{ij}-\mu_j)}{\hat \sigma_{j}^{\tau} n^{1/2}},
 \end{equation}
 where $\hat t_j=\sqrt n \hat\sigma_j$ for $j=1,\dots,p$, and $\tau\in[0,1]$ is a fixed partial standardization parameter. 

 Importantly, there is a direct link between distributional approximations for $\mathcal{M}_n$ and inference on the means $\mu_1,\dots,\mu_p$. This is due to the monotonicity of the functions $\varphi_{\hat t_j}(\cdot)$, which makes it straightforward to construct simultaneous confidence intervals for the means using quantile estimates for $\mathcal{M}_n$ and its corresponding min statistic, as discussed in Section~\ref{sec:method}. The robust variance estimates $\hat\sigma_j^2$ also play an essential role, because they ensure that the confidence intervals induced by $\mathcal{M}_n$ are automatically adapted to the scale of the covariates, which is an issue that has often been neglected in the literature on max statistics. 

For the purpose of bootstrapping $\mathcal{M}_n$, let $\tilde X_j$ denote a hold-out \textsc{mom} estimate of $\mu_j$ to be defined later, and let $\bar \varphi_j=\frac{1}{n}\sum_{i=1}^n \varphi_{\hat t_j}(X_{ij}-\tilde X_j)$. In addition, let $\xi_1,\dots,\xi_n\sim N(0,1)$ be i.i.d.~Gaussian multipliers generated independently of the data. Putting these pieces together, we define a bootstrap sample of $\mathcal{M}_n$ as
\begin{equation}\label{eqn:bootdef}
\small
\mathcal{M}_n^*=\max_{1\leq j\leq p}\sum_{i=1}^n  \frac{\xi_i\big(\varphi_{\hat t_j}(X_{ij}-\tilde X_j)-\bar \varphi_j\,\big)}{\hat \sigma_{j}^{\tau}n^{1/2}}.
\end{equation}

With regard to theoretical analysis, we focus on a setting where the tails of the data are quantified by a variant of the  $L^4$-$L^2$ moment equivalence condition, which has gained increasing currency in the high-dimensional robustness literature~\cite{lugosi2019mean, ke2019user,mendelson2020robust, roy2021empirical, 
abdalla2022covariance}. Specifically, we assume there is some $\delta>0$ such that the bound $\|\langle v,X_1-\mu\rangle\|_{L^{4+\delta}}\lesssim \|\langle v,X_1-\mu\rangle\|_{L^2}$ holds for all $v\in\R^p$, and in Proposition~\ref{prop:moments}, we show that this condition is satisfied by heavy-tailed instances of well-known models. The other primary structural assumption in our analysis is that the covariates have a weak variance decay property of the form $\sigma_{(j)}^2\asymp j^{-2\beta}$ for some fixed $\beta>0$, where $\sigma_{(1)}^2\geq \cdots\geq \sigma_{(p)}^2$ are the sorted coordinate-wise variances. Notably, the decay is referred to as weak because the parameter $\beta$ is allowed to be \emph{arbitrarily small}. Furthermore, it is known that this type of structure arises naturally in a variety of high-dimensional contexts that are related to principal components analysis and functional data analysis, among others~\citep{lopes2020bootstrapping}. 
Under the complete set of conditions given in Assumption~\ref{A:model}, our main result shows that with high probability, the Kolmogorov distance $\sup_{s\in \R}|\P(\mathcal{M}_n\leq s)-\P(\mathcal{M}_n^*\leq s|X)|$ is nearly of order $n^{-1/2}$, where $\P(\cdot|X)$ denotes probability that is conditional on all of the observations. 

From a practical standpoint, the proposed method has several strengths. First, the method does not require fine tuning, which is demonstrated by the fact that we use the simple choices of $\hat t_j=\sqrt n \hat\sigma_j$ and $\tau=0.9$ throughout all of the experiments presented in Section~\ref{sec: simu}. Second, we show that the method reliably produces well-calibrated tests and confidence intervals across many conditions---including heavy-tailed data generated from separable and elliptical distributions, as well as heavy-tailed functional data with rough sample paths. Third, the simulation results reveal that the proposed method performs favorably in comparison to the pseudomedian approach in~\citep{Fan:2023}.\\

\noindent\textbf{Notation.} If $A$ is a real matrix, its Frobenius norm is $\|A\|_F=\sqrt{\textup{tr}(A\ttop A)}$, and its operator norm $\|A\|_{\textup{op}}$ is the same as its largest singular value. If $x$ and $y$ are Euclidean vectors of the same dimension, then $\langle x,y\rangle$ denotes the Euclidean inner product, and $\|x\|_2=\sqrt{\langle x,x\rangle}$. If $\xi$ is a scalar random variable and $1\leq q<\infty$, we write $\|\xi\|_{L^q}=(\E|\xi|^q)^{1/q}$, and in the case when $q=\infty$, we use $\|\xi\|_{L^{\infty}}$ to refer to the essential supremum. If $f$ is a scalar-valued function on $\R$, the notation $\|f\|_{L^{\infty}}$ is understood analogously with respect to Lebesgue measure.  If $\{a_n\}$ and $\{b_n\}$ are sequences of non-negative real numbers, then the relations $a_n \lesssim b_n$ and $a_n = \mathcal{O}(b_n)$ are equivalent, and mean that there is a constant $c>0$ not depending on $n$, such that $a_n \leq cb_n$ holds for all large $n$.  If $a_n\lesssim b_n$ and $b_n\lesssim a_n$ both hold, then we write $a_n\asymp b_n$. Lastly, let $a_n \vee b_n=\max\{a_n,b_n\}$.

\section{Method}
\label{sec:method}
Here, we provide the details for constructing the bootstrap sample $\mathcal{M}_n^*$, as well as simultaneous confidence intervals $\hat{\mathcal{I}}_1,\dots,\hat{\mathcal{I}}_p$ for the coordinate-wise means $\mu_1,\dots,\mu_p$. Further applications of these intervals to various testing problems will be covered later in Section~\ref{sec: simu}.

In addition to the observations $X_1,\dots,X_n$ discussed above, let $X_{n+1},\dots,X_{n+m_n}$ denote an independent set of i.i.d.~hold-out observations generated from the same distribution. For simplicity, the number of hold-out observations $m_n$ is assumed to be even, and in all of our numerical experiments, we will take $m_n$ to be about 10\% of $n$. The hold-out observations are used to construct robust estimators $\tilde X_1,\dots,\tilde X_p$ and $\hat\sigma_1^2,\dots,\hat\sigma_p^2$ for the coordinate-wise means and variances, which are the only ingredients for generating $\mathcal{M}_n^*$ that were not addressed previously in Section~\ref{sec:intro}. Taking an \textsc{mom} approach, we partition the hold-out indices $\{n+1,\dots,n+m_n\}$ into $b_n$ blocks $\mathcal{B}_1,\dots,\mathcal{B}_{b_n}$, with each block containing an even number of $\ell_n$ indices such that $m_n=\ell_n b_n$. More specifically, let $\mathcal{B}_1=\{n+1,\dots,n+\ell_n\}$, $\mathcal{B}_2=\{n+\ell_n+1,\dots,n+2\ell_n\}$, and so on. 
For the $l$th block, let 
\begin{equation}
    \bar X_j(l)=\frac{1}{\ell_n}\sum_{i\in \mathcal{B}_l} X_{ij}
\end{equation}
 denote the block-wise sample mean of the $j$th coordinate, and define the \textsc{mom} estimator of $\mu_j$ as
\begin{equation}
\small
    \tilde X_j = \textup{median}(\bar X_j(1),\dots,\bar X_j(b_n)).
\end{equation}
Likewise, we construct an \textsc{mom} estimate for each $\sigma_j^2$ along similar lines. The $l$th blockwise estimate for $\sigma_j^2$ is obtained by averaging the squared differences of $\ell_n/2$ pairs of observations
\begin{equation}\label{eqn:barsigmadef}
\small
    \bar\sigma_j^2(l)
    = \frac{1}{\ell_n/2}\sum_{ \substack{i,i'\in \mathcal{B}_l\\ i'-i=\ell_n/2}} \frac{1}{2}(X_{ij}-X_{i'j})^2,
\end{equation}
and then $\hat\sigma_j^2$ is defined to be the median of the block-wise estimates
\begin{equation}
\small
\hat\sigma_j^2=\textup{median}(\bar\sigma_{j}^2(1),\dots,\bar\sigma_{j}^2(b_n)).
\end{equation}

Next, we turn to the construction of simultaneous confidence intervals $\hat{\mathcal{I}}_1,\dots,\hat{\mathcal{I}}_p$ for $\mu_1,\dots,\mu_p$.  Let $1-\alpha$ denote the nominal simultaneous coverage probability, and let \smash{$\hat q_+(1-\alpha/2)$} denote the empirical $(1-\alpha/2)$-quantile of a collection of bootstrap samples generated in the manner of $\mathcal{M}_n^*$. Also, let $\underline{\mathcal{M}}_n^*$ denote the counterpart of $\mathcal{M}_n^*$ that is obtained by replacing $\max_{1\leq j\leq p}$ with $\min_{1\leq j\leq p}$ in equation~\eqref{eqn:bootdef}, and let $\hat q_{-}(\alpha/2)$ denote the empirical $(\alpha/2)$-quantile of a collection of bootstrap samples generated in the manner of $\underline{\mathcal{M}}_n^*$. In this notation, the confidence interval $\hat{\mathcal{I}}_j$ is defined by
\begin{equation}
\label{equ:Ij}
\hat{\mathcal{I}}_j=\bigg\{x\in\R \, : \, \hat q_{-}(\alpha/2)\leq \frac{1}{ \sqrt n\hat \sigma_{j}^{\tau}}\sum_{i=1}^n \varphi_{\hat t_j}(X_{ij}-x)\leq \hat q_{+}(1-\alpha/2)\bigg\}.
\end{equation}
Due to the fact that the functions $\varphi_{\hat t_1}(\cdot),\dots,\varphi_{\hat t_p}(\cdot)$ are monotone, it is straightforward to compute all the endpoints of $\hat{\mathcal{I}}_1,\dots,\hat{\mathcal{I}}_p$.

To comment on the role of the partial standardization parameter $\tau\in[0,1]$, it provides a way to balance two opposing effects that occur in the extreme cases when $\tau$ is equal to 0 or 1. When $\tau=0$, all of the intervals $\hat{\mathcal{I}}_1,\dots,\hat{\mathcal{I}}_p$ have the same width, which is clearly undesirable when the covariates fluctuate over different scales. Alternatively, when $\tau=1$, all of the covariates will be on approximately ``equal footing'', which will tend to make the max statistic $\mathcal{M}_n$ sensitive to all $p$ dimensions. This is undesirable in high-dimensional situations where the covariates fluctuate over different scales, because it eliminates a form of low-dimensional structure that can simplify the behavior of $\mathcal{M}_n$ when $\tau<1$. To see this, consider a case where $\tau=0$ and  $\sigma_1,\dots,\sigma_d$ are much larger than $\sigma_{d+1},\dots,\sigma_p$ for some $d\ll p$. In this case, the maximizing index for $\mathcal{M}_n$ is likely to reside in the small subset $\{1,\dots,d\}\subset\{1,\dots,p\}$. Thus, the behavior of $\mathcal{M}_n$ will be mainly governed by the first $d$ covariates, which intuitively reduces the effective dimension of the problem of approximating the distribution of $\mathcal{M}_n$. Accordingly, as was originally proposed in~\citep{lopes2020bootstrapping}, it is natural to select an intermediate value of $\tau$ between 0 and 1 that can mitigate the unwanted effects that occur at $\tau\in\{0,1\}$.

\section{Theory}
\label{sec:theory}
Our theoretical analysis is framed in terms of a sequence of models that are implicitly embedded in a triangular array whose rows are indexed by $n$. In this context, all model parameters are allowed to vary with $n$, except when stated otherwise. In particular, the dimension $p=p(n)$ is regarded as a function of $n$, and hence, if a quantity does not depend on $n$, then it does not depend on $p$ either. 

To state our model assumptions, recall that the sorted coordinate-wise variances of $X_1$ are denoted as $\sigma_{(1)}^2\geq \cdots\geq \sigma_{(p)}^2$, and for any $d\in \{1,\dots,p\}$, let $J(d)$ be a set of $d$ indices in $\{1,\dots,p\}$ that satisfies $\{\sigma_j^2\,|\,j\in J(d)\}=\{\sigma_{(1)}^2,\dots,\sigma_{(d)}^2\}$. In addition, let $R(d)$ denote the $d \times d$ correlation matrix associated with the covariates $\{X_{1j}\}_{j \in J(d)}$.

\begin{assume}\label{A:model}{} The observations $X_1,\dots,X_{n+m_n}\in\R^p$ are i.i.d., and there are constants $C\geq 1$, $\beta>0$, and $\delta \geq \epsilon>0$ not depending on $n$ such that the following conditions hold:
\begin{enumerate}[label=(\roman*)]

\item \label{A:moments}
For all $v\in\R^p$, $\big\|\big\langle v,X_1 -\E(X_1)\big\rangle\big\|_{L^{4+\delta}} \, \leq \,  C \big\|\big \langle v,X_1 -\E(X_1)\big\rangle\big\|_{L^{2}}$ holds.

\item \label{A:conti}
For all $j=1,\dots,p$, the random variable $X_{1j}/\sigma_j$ has a Lebesgue density $f_j$ such that $\|f_j\|_{L^{\infty}}\leq C$.

\item \label{A:var}
For all $j=1,\dots,p$, the inequalities $\frac{1}{C}\sigma_{(1)}^2 j^{-2\beta}\leq \sigma_{(j)}^2\leq C \sigma_{(1)}^2 j^{-2\beta}$ hold.
\item
\label{A:cor}
If $l_n=\big\lceil n^{\frac{\epsilon}{24(\beta \vee 1)}}\!\wedge p\big\rceil$, then $\|R(l_n)\|_F^2\leq C l_n^{2-\frac{1}{C}} $. 
\end{enumerate}
\end{assume}
\noindent\textbf{Remarks.} All of the conditions in Assumption~\ref{A:model} are invariant to shifting $X_1\mapsto X_1+v$ for fixed $v\in\R^p$, and scaling $X_1\mapsto c X_1$  for fixed $c\neq 0$. The following paragraphs provide several examples that address each of the conditions (\ref{A:moments})-(\ref{A:cor}). Also, it is straightforward to combine the examples to construct a wide assortment of data-generating distributions that satisfy all of the conditions in Assumption~\ref{A:model} simultaneously.\\

\noindent\textbf{Examples of $L^{4+\delta}$-$L^2$ moment equivalence.} In recent years, moment assumptions similar to
condition (\ref{A:moments}) have been adopted in many analyses of robust statistical methods for high-dimensional data~\cite{lugosi2019mean, ke2019user,mendelson2020robust, roy2021empirical, 
abdalla2022covariance}. As shown in Proposition~\ref{prop:moments} below, condition~(\ref{A:moments}) is compatible with the classes of \emph{elliptical} and \emph{separable} models (also known as independent component models), which are widely used in areas such as multivariate analysis, random matrix theory, and signal processing~\citep{kotz2019continuous,Bai:2010,Comon:2010}.

To be precise, we say that an observation $X_1$ with mean $\mu$ and covariance matrix $\Sigma$ has an elliptical distribution if it can be represented as $X_1=\mu+\eta_1 \Sigma^{1/2} U_1$, where $U_1\in\R^p$ is uniformly distributed on the unit sphere, and $\eta_1$ is a non-negative scalar random variable that is independent of $U_1$ and normalized by $\E(\eta_1^2)=p$. On the other hand, we say that $X_1$ has a separable distribution if it can be represented as $X_1=\mu+\Sigma^{1/2}\zeta_1$, where $\zeta_1=(\zeta_{11},\dots,\zeta_{1p})$ has i.i.d.~entries with $\E(\zeta_{11})=0$, and $\Var(\zeta_{11})=1$.\label{modelpage}
\begin{proposition}\label{prop:moments} Conditions \ref{A:moments} and~\ref{A:conti} hold simultaneously if one of the following two conditions holds for some $\delta>0$ that is fixed with respect \smash{to $n$.}
\begin{itemize}
    \item[(1)] The observation $X_1$ is drawn from an elliptical distribution such that \smash{$\|\eta_1\|_{L^{4+\delta}} \lesssim  \sqrt{p}$}, and the random variable $X_{11}/\sigma_1$ has a Lebesgue density $f_1$  such that $\|f_1\|_{L^{\infty}}\lesssim 1$. 

    \item[(2)] The observation $X_1$ is drawn from a separable distribution with\\ \smash{$\max_{1\leq j\leq p}\|\zeta_{1j}\|_{L^{4+\delta}}$ $\lesssim$ $1$,} and each random variable $\zeta_{1j}$ has a Lebesgue density $g_j$ such that $\max_{1\leq j\leq p}\|g_j\|_{L^{\infty}}\lesssim 1$.
\end{itemize}
\end{proposition}
\noindent The proof is provided in Appendix G.\\

\noindent\textbf{Examples of variance decay.} There are a variety of settings where the sorted coordinate-wise variances $\sigma_{(1)}^2\geq \cdots \geq \sigma_{(p)}^2$ naturally exhibit a decay profile.\\

\noindent\emph{Principal components analysis}. In the context of principal components analysis, it is common to assume that the sorted eigenvalues $\lambda_1(\Sigma)\geq\cdots\geq \lambda_p(\Sigma)$ of $\Sigma$ satisfy $\lambda_j(\Sigma)\lesssim j^{-\gamma}$ for some $\gamma>0$, and it can be shown that this implies $\sigma_{(j)}^2\lesssim j^{-2\beta}$ for some other decay parameter $\beta>0$~\citep[][Proposition 2.1]{lopes2020bootstrapping}.\\

\noindent\emph{Mean-variance proportionality}.  Another scenario where variance decay arises is when the coordinate-wise means and variances are connected by a proportionality relationship of the form $\sigma_j^2\propto |\mu_j|^{\gamma}$, for some fixed exponent $\gamma>0$. This occurs within many sub-families of distributions, including Gamma, Weibull, inverse Gaussian, and Pareto. In applications that involve sparse  modelling of high-dimensional mean vectors, a classical assumption is that the sorted coordinate-wise means have a decay profile~\citep{Johnstone}, and thus, when such a proportionality relationship  holds, it follows that variance decay must also occur.\\

\noindent\emph{Functional data analysis}. One more set of examples is related to functional data analysis, where function-valued observations $\Psi_1,\dots,\Psi_n$ in a Hilbert space are often studied through their projections under a finite number of orthonormal basis functions $\{\phi_j\}_{1\leq j\leq p}$. That is, the $i$th projected observation has the form $X_i=(\langle \Psi_i,\phi_1\rangle,\dots,\langle\Psi_i,\phi_p\rangle)\in\R^p$. In connection with our work, the important point is that under standard assumptions in functional data analysis, it can be shown that the sorted coordinate-wise variances of $X_i$ have a decay profile~\citep{lopes2020bootstrapping}. In fact, this occurs even when the random functions $\Psi_1,\dots,\Psi_n$ have rough sample paths, which we illustrate empirically in Figure~\ref{fig:simu apple sd}.\\

\noindent\textbf{Examples of correlation matrices.} To interpret the condition (\ref{A:cor}), it should be noted that the inequality $\|R(l_n)\|_F^2 \leq l_n^2$ always holds, since $\|A\|_F^2\leq \textup{tr}(A)^2$ holds for any positive semidefinite matrix $A$. So, in this sense, condition (\ref{A:cor}) is quite mild, as $C$ may be taken to be arbitrarily large. Moreover, the correlation structure of the variables indexed by $\{1,\dots,p\}\setminus J(l_n)$ is \emph{completely unrestricted}. With regard to the constant 24 appearing in the definition of $l_n$, it has no special importance, and is used for theoretical convenience.  Below, we describe several classes of $p\times p$ correlation matrices $R=R(p)$ for which the sub-matrix $R(l_n)$ satisfies condition (\ref{A:cor}). \\

\noindent \emph{Decaying correlation functions}. Let $\rho: [0,\infty) \rightarrow [0,1]$ be any continuous convex function satisfying $\rho(0)=1$, and $\rho(t) \leq c t^{-\gamma}$ for some fixed constants $c>0, \gamma>0$, and all $t \geq 0$. By P\'{o}lya's criterion \citep{polya1949remarks}, a matrix whose $ij$ entry is defined by $\rho(|i-j|)$ is a correlation matrix that satisfies condition (\ref{A:cor}). Correlation matrices of this type include many well-known examples, such as those of the autoregressive and banded types, e.g.,~ $R_{ij}=r^{|i-j|}$ for some fixed $r\in (0,1)$, and $R_{ij}=\max \big\{0, 1-\frac{|i-j|}{b}\big\}$ for some fixed $b>0$.\\

\noindent\emph{Diverging operator norm}. If the operator norm of $R$ satisfies $\|R\|_{\textup{op}}\leq C l_n^{\frac{1}{2}-\frac{1}{2C}}$, then Assumption~\ref{A:model}(\ref{A:cor}) holds. This can be seen by noting that $\|R(l_n)\|_F^2\leq l_n\|R(l_n)\|_{\textup{op}}^2\leq l_n\|R\|_{\textup{op}}^2$. In particular, since $l_n$ increases with $n$ and $p$,  this shows that condition (\ref{A:cor}) can hold even when the operator norm of $R$ diverges asymptotically.\\

\noindent \emph{Block structure}. Suppose that $R$ is formed by concatenating $k$ correlation matrices along its diagonal, with sizes $\nu_1\times \nu_1,\dots,\nu_k\times \nu_k$, so that $\nu_1+\cdots+\nu_k=p$. If the condition  $\max\{\nu_1,\dots,\nu_k\}\leq C l_n^{1-\frac{1}{C}}$ holds, then so does condition (\ref{A:cor}). This follows from the observation that no row of $R$ can have a squared $\ell_2$ norm larger than $\max\{\nu_1,\dots,\nu_k\}$, and so $\|R(l_n)\|_F^2\leq l_n \max\{\nu_1,\dots,\nu_k\}$.\\

\noindent \emph{Convex combinations and permutations}. If $R$ and $R'$ denote any correlation matrices corresponding to the previous examples, then for any $t\in[0,1]$, the correlation matrix $t R+(1-t)R'$ satisfies condition (\ref{A:cor}). Furthermore, if $\Pi$ is a $p\times p$ permutation matrix, then $\Pi R \Pi\ttop$ is also a correlation matrix that satisfies condition (\ref{A:cor}). These operations considerably extend the examples that have decaying correlation functions or block structure.\\

The following theorem is our main result.

\begin{theorem}\label{thm:main}
   Fix any constant $\tau\in[0,1)$ with respect to $n$, and suppose that Assumption~\ref{A:model} holds with the values of $\delta\geq\epsilon>0$ stated there. In addition, suppose that the hold-out set consists of $m_n\asymp n$ observations that are partitioned into $b_n\asymp \log(n)$ blocks. Then, there is a constant $c> 0$ not depending on $n$ such that the event
\begin{equation*}
    \sup_{s\in\R}\Big| \P(\mathcal{M}_n\leq s) - \P(\mathcal{M}_n^*\leq s|X)\Big| \ \leq \ c\ \!n^{-\frac{1}{2}+\epsilon}
\end{equation*}
occurs with probability at least $1-c\ \! n^{-\delta/4}$.
\end{theorem}
\noindent\textbf{Remarks.} In addition to the fact that the rate of approximation is near $n^{-1/2}$, it should be noted that the rate does not depend on the dimension $p$ or on the size of the variance decay parameter $\beta$. A key step in the proof is to ``localize'' the random maximizing index, say $\hat\j\in\{1,\dots,p\}$, that satisfies $\frac{1}{\hat\sigma_{\hat \j}^{\tau}n^{1/2}}\sum_{i=1}^n (\varphi_{\hat t_{\hat \j}}(X_{i\hat \j})-\mu_{\hat \j})=\mathcal{M}_n$. This involves showing that there is a non-random set $A\subset \{1,\dots,p\}$ with cardinality $|A|\ll p$ such that $\hat \j$ falls into $A$ with high probability. Consequently, the max statistic $\mathcal{M}_n$ can be analyzed as if the data reside in the low-dimensional space $\R^{|A|}$. In the proofs of Proposition \ref{prop: piece2} and Lemma \ref{Lemma: tilde MX}, we implement this strategy using a key technical tool (Lemma~\ref{lem:lowertail}), which  is a lower-tail bound for the maximum of correlated Gaussian variables, established in~\citep{lopes2022sharp}.

The proof also analyzes several effects that arise from heavy-tailed covariates and the structure of the robust max statistic $\mathcal{M}_n$. A particularly important example of such an effect is the bias that is introduced by the truncation functions $\varphi_{\hat t_j}$, and this is addressed in the proofs of Propositions \ref{Prop: practical kn} and \ref{prop: XZ kn-p}, as well as Lemmas ~\ref{Lemma: hat Y and Y Mkn}, ~\ref{Lemma: Y hatY diff}, and \ref{lemma: XZ kn-kn}. Furthermore, our work accounts for the fluctuations of the robust \textsc{mom} estimates $\tilde X_j$ and $\hat\sigma_j$ that are used to partially standardize $\mathcal{M}_n^*$, and this is done in Lemmas \ref{Lemma: MOM mean}-\ref{Lemma: MOM var anti-con}.

\section{Numerical experiments}
\label{sec: simu}
This section addresses the practical performance of the proposed method in the contexts of Euclidean and functional data with heavy tails. In addition, we include the performance comparisons with the method for robust inference proposed in~\citep{Fan:2023}.

\subsection{Euclidean data}
\label{sec: euclidean simu}
Here, we consider the task of constructing simultaneous confidence intervals for the entries of the mean vector $(\mu_1,\dots,\mu_p)=\E(X_1)$, based on i.i.d.~observations $X_1,\dots,X_{n+m_n}\in\R^p$.\\

\noindent\textbf{Design of experiments.} For each pair $(n+m_n,p)$ in the set $\{500\}\times \{100,500,1000\}$, and each of the data-generating distributions described below, we performed 500 Monte Carlo trials. 
Due to the fact that the proposed method and the method in~\citep{Fan:2023} are shift invariant, the mean vector $(\mu_1,\dots,\mu_p)$ was always chosen to be the zero vector without loss of generality.
 In all trials, the proposed confidence intervals $\hat{\mathcal{I}}_1,\dots,\hat{\mathcal{I}}_p$ were constructed according to~\eqref{equ:Ij}, using 500 bootstrap samples, and using choices of 90\% and 95\% for the nominal simultaneous coverage probability $1-\alpha$. Also, the proposed method was always applied with truncation parameters set to $\hat t_j=\sqrt n\hat\sigma_j$ for all $j=1,\dots,p$, and with the partial standardization parameter set to $\tau=0.9$. Lastly, in all trials, the number of hold-out observations was set to $m_n=50$, and the block length for the \textsc{mom} estimates was set to $\ell_n=10$.\\

\noindent\textbf{Data-generating distributions.} The data were generated in four ways, based on an elliptical distribution and a separable distribution with two choices of covariance matrices.\\

\noindent\emph{Elliptical distribution.} The elliptical observations have the form $X_1=\mu+\eta_1\Sigma^{1/2}U_1$, where $U_1\in\R^p$ is uniformly distributed on the unit sphere, and $\eta_1\geq 0$ is a random variable that is independent of $U_1$ such that $3\eta_1^2/(2p)$ follows an $F$ distribution with $p$ and 6 degrees of freedom. This distribution for $X_1$ is more commonly known as a \emph{multivariate $t$-distribution on 6 degrees of freedom}~\citep{muirhead2009aspects}.\\

\noindent\emph{Separable distribution.} The separable observations have the form $X_1=\mu+\Sigma^{1/2}\zeta_1$, where $\zeta_1=(\zeta_{11},\dots,\zeta_{1p})$ has i.i.d.~entries that are standardized Pareto random variables. Specifically, $\zeta_{11}=(\omega_{11}-\E(\omega_{11}))/\sqrt{\var(\omega_{11})}$, where $\omega_{11}$ is drawn from a Pareto distribution whose density is given by $x\mapsto 6x^{-7} 1\{ x \geq 1\}$.\\

\noindent\emph{Covariance matrices.} For both the elliptical and separable distributions, we constructed the covariance matrix of $X_1$ in the form $\Sigma=D^{1/2}RD^{1/2}$, where $D=\textup{diag}(\var(X_{11}),\dots,\var(X_{1p}))$, and $R$ is the correlation matrix of $X_1$. The correlation matrix was chosen to be one of the following two types
\begin{equation*}
\small
    R_{ij}=\begin{cases}0.5^{|i-j|} \ \ \ \ \  \ \ \ \ \ \ \ \ \ \   \ \ \ \ \ \ \ \ \ \ \  \ \ \text{(autoregressive)}\\
   \ts1\{i=j\}+\frac{1\{i\neq j\}}{4(i-j)^2} \ \  \ \ \ \ \ \ \ \ \ \ \ \ \text{(algebraic decay).}
   \end{cases}
\end{equation*}
In all cases, the matrix $D$ was chosen to satisfy $D_{jj}^{1/2}=j^{-1/2}$ for all $j=1,\dots,p$, ensuring that the entries of $X_1$ have variance decay.\\

\noindent\textbf{Discussion of results.} In Table~\ref{T:simu}, we report empirical simultaneous coverage probabilities and width measures, for both the proposed method (denoted PM) and the method based on the Hodges-Lehmann estimator (denoted HL) developed in~\citep{Fan:2023}. The simultaneous coverage probabilities were computed as the fraction of the 500 Monte Carlo trials in which all $p$ intervals of a given method contained the corresponding parameters $\mu_1,\dots,\mu_p$. The width measure was computed as the median width of the $p$ intervals, averaged over the 500 trials, and it is shown in parentheses below the simultaneous coverage probabilities. 

Across all of the settings, PM delivers simultaneous coverage probabilities that are within about 1 or 2 percent of the nominal level, demonstrating that it is reliably calibrated. On the other hand, HL is only well calibrated in the cases of elliptical models, where the covariates have distributions that are symmetric around 0. In the cases of separable models where the covariates have distributions that are \emph{not} symmetric around 0, the simultaneous coverage probabilities of HL are far from the nominal level. This can be explained by the fact that HL is designed for simultaneous inference on the coordinate-wise pseudomedians---which may differ from the coordinate-wise means if the covariates have asymmetric distributions.
By contrast, PM has a scope of application for inference on coordinate-wise means that is not limited by asymmetric distributions.

With regard to the width measure, Table~\ref{T:simu} shows that the intervals produced by PM are much tighter than those produced by HL, sometimes by a factor of 4 or more. This is to be expected, because the HL intervals have equal widths across all $p$ coordinates, whereas PM adapts to the scale of each covariate and takes advantage of variance decay. This difference between the methods is also reflected in another pattern---which is that the width measure for PM decreases as $p$ increases, whereas the width measure for HL stays essentially constant as $p$ increases, since HL uses unstandardized covariates.

\begin{table}
\centering
\caption{Comparison of simultaneous coverage probability and confidence interval width: PM refers to the proposed method, and HL refers to the method based on the Hodges-Lehmann estimator proposed in~\citep{Fan:2023}.} 
\label{T:simu}
\begin{tabular}{c | c| c |c c |cc|cc}
\hline
\hline
\multirow{2}{*}{$R$} & \multirow{2}{*}{\text{Distribution}} & \multirow{2}{*}{$\alpha$} & \multicolumn{2}{|c|}{\( p = 100 \)} & \multicolumn{2}{c|}{\( p = 500 \)} & \multicolumn{2}{c}{\( p = 1000 \)}\\ \cline{4-9} 
&  &  & PM & HL & PM & HL & PM & HL \\ \hline
\multirow{8}{*}{\shortstack{auto-\\regressive}} & \multirow{4}{*}{elliptical}
& \multirow{2}{*}{0.05}  & 0.942 & 0.942 & 0.956 & 0.938 & 0.962 & 0.962\\
& &  &  (0.057) & (0.165) & (0.031) & (0.165) & (0.023) & (0.166) \\
& & \multirow{2}{*}{0.1} & 0.914 & 0.892 & 0.918 & 0.900 & 0.908 & 0.916\\
& &  & (0.053) & (0.142) & (0.029) & (0.142) & (0.022) & (0.142) \\ \cline{2-9}
 & \multirow{4}{*}{separable}
 & \multirow{2}{*}{0.05}  & 0.956 & 0.006 & 0.946 & 0.002 & 0.954 & 0\\
& &  &  (0.065) & (0.143) & (0.035) & (0.144) & (0.027) & (0.143)\\
& & \multirow{2}{*}{0.1} & 0.922  & 0 & 0.904 & 0.002 & 0.910 & 0 \\
& &  &  (0.059)  & (0.123) & (0.033) & (0.124) & (0.025) & (0.123)\\ \hline
\multirow{8}{*}{\shortstack{algebraic\\decay}} & \multirow{4}{*}{elliptical}
& \multirow{2}{*}{0.05}  & 0.946 & 0.938 & 0.942 & 0.942 & 0.958 & 0.962\\
& &  &  (0.058) & (0.167) & (0.031) & (0.166) & (0.023) & (0.166) \\
& & \multirow{2}{*}{0.1} & 0.914 & 0.890 & 0.910 & 0.898 & 0.916 & 0.904\\
& &  & (0.053) & (0.144) & (0.029) & (0.143) & (0.022) & (0.144) \\ \cline{2-9}
 & \multirow{4}{*}{separable}
 & \multirow{2}{*}{0.05}  & 0.968 & 0 & 0.954 & 0 & 0.958 & 0\\
& &  & (0.067) & (0.141) & (0.037) & (0.142) & (0.028) & (0.141)\\
& & \multirow{2}{*}{0.1} &  0.922 & 0 & 0.910 & 0 & 0.912 & 0\\
& &  & (0.062) & (0.122) & (0.034) & (0.123) & (0.026) & (0.122) \\
\hline \hline
\end{tabular}
\end{table}

\subsection{Functional data}
\label{sec: simu bg}

 Beyond heavy-tailed Euclidean data, our approach to inference with robust max statistics can be applied to heavy-tailed functional data. Below, we study the problem of detecting a non-zero drift in functional observations that arise from geometric Brownian motion (GBM). This is a heavy-tailed stochastic process in the sense that its marginals are \smash{lognormal,} which is a standard example of heavy-tailed univariate distribution~\citep{Foss,nair}. The sample paths of GBM present additional challenges from the standpoint of functional data analysis, because they are rough. Furthermore, GBM is of broad interest in financial applications, where it is widely used for modelling securities prices~\citep{Oksendal}.\\

\noindent \textbf{Problem formulation.} A sample path of GBM on the unit interval has the form $t\mapsto \exp\big( (\mu_0-\varsigma_0^2/2)t+\varsigma_0 W(t)\big)$ for $t\in[0,1]$,
where $W(t)$ is a standard Brownian motion, $\mu_0\in\R$ is the drift parameter, and $\varsigma_0^2\geq 0$ is the volatility parameter. To formalize the detection of non-zero drift in a way that allows us to incorporate natural alternative hypotheses, we will allow for more general sample paths $S(t)$ of the form
\begin{align}\label{eqn:GBM}
\small
S(t)=\exp\big( (h \mu(t)-\varsigma_0^2/2)t+\varsigma_0 W(t)\big),
\end{align}
where $\mu(t)$ is a fixed real-valued function on $[0,1]$ such that $ S(t)$ resides in $L^2[0,1]$ almost surely, and $h\geq 0$ is a fixed parameter that measures the ``distance'' from the null hypothesis of zero drift that occurs when $h=0$. The parameters $\mu(t)$ and $\varsigma_0$ are treated as unknown. Under these conditions, we are interested in using a dataset $S_1(t),\dots,S_{n+m_n}(t)$ of i.i.d.~samples  of $S(t)$ to address the hypothesis testing problem
\begin{equation}
\small
    \mathsf{H}_0: h=0 \text{ \ \ \ vs. \ \ \  } \mathsf{H}_1: h>0.
\end{equation}
In particular, different choices of the function $\mu(t)$ correspond to different alternatives, and later on, we will present numerical results for several choices.\\

\noindent \textbf{Testing procedure.} The sample path formula~\eqref{eqn:GBM} implies that $\E(S(t))=\exp(h\mu(t)t)$ for all $t\in[0,1]$. For this reason, our procedure will seek to detect whether or not the function $\E(S(t))-1$ is identically 0. This will be done by expanding \smash{$\E(S(t))-1$} in the form $\sum_{j=1}^{\infty} \beta_j \phi_j(t)$, where $\{\phi_j(t)\}_{j\geq 1}$ is the Fourier cosine basis for $L^2[0,1]$. 
To proceed, note that $\E(S(t))-1$ is equal to the zero function in the $L^2[0,1]$ sense if and only if $\beta_j=0$ for all $j\geq 1$. This motivates a procedure based on testing the simultaneous hypotheses
\begin{equation}
\small
    \mathsf{H}_{0,j}:\beta_j=0 \text{ \ \ \ for \ \ \ } j=1,\dots,p,
\end{equation}
where $p$ is an integer large enough so that the coefficients $\beta_{p+1},\beta_{p+2},\dots,$ are negligible for practical purposes. In particular,  $\mathsf{H}_0$ implies that $\mathsf{H}_{0,1},\dots,\mathsf{H}_{0,p}$ hold simultaneously, and so a procedure that controls the simultaneous type I error rate for these hypotheses leads to one that controls the ordinary type I error rate for $\mathsf{H}_0$. 

If we define $X_i\in\R^{p}$ to contain the first $p$ coefficients of $S_i(t)-1$ with respect to $\{\phi_j(t)\}_{j\geq 1}$, then it follows that $\E(X_i)=(\beta_1,\dots,\beta_p)$ for every $i=1,\dots,n+m_n$. Hence, we may apply our proposed method from Section~\ref{sec:method} to $X_1,\dots,X_{n+m_n}$ in order to construct confidence intervals $\hat{\mathcal{I}}_1,\dots,\hat{\mathcal{I}}_p$ for $\beta_1,\dots,\beta_p$ with a nominal simultaneous coverage probability of $1-\alpha$. Altogether, this means that if we reject $\mathsf{H}_0$ when any of these intervals exclude 0, then this rejection rule corresponds to a test with a nominal level of at most $\alpha$.\\

\noindent\textbf{Design of experiments.} To construct four natural choices of the pair $(\mu(t),\varsigma_0)$, we used historical data for the stocks of Apple, Nvidia, Moderna, and JPMorgan, to be described later in the paragraph labeled `preparation of stock data'. With regard to the choice of $\varsigma_0^2$, note that if a stock price is modeled as a realization of $S(t)$, then the pointwise variance of the cumulative log return is $\var((h\mu(t)-\varsigma_0^2/2)t+\varsigma_0 W(t))= \varsigma_0^2 t$, and the time average of this quantity over the unit interval is $\varsigma_0^2/2$. (See~\citep{ruppertstatistics} for additional background.) Accordingly, for each of the four stocks, we selected $\varsigma_0$ such that $\varsigma_0^2/2=\int_0^1 s^2(t)\,\mathrm{d}t$, where $s^2(t)$ is the sample pointwise variance of the cumulative log return of the stock over many disjoint periods of unit length. Next, the selection of $\mu(t)$ was motivated by the fact that if a stock price is modeled with $S(t)$, and if $h=1$, then the pointwise expected cumulative log return is $(\mu(t)-\varsigma_0^2/2)t$. For a given value of $\varsigma_0$, we defined $\mu(t)$ to satisfy $(\mu(t)-\varsigma_0^2/2)t=\bar R(t)$, where $\bar R(t)$ is the sample average of the cumulative log return curves of a given stock over many disjoint periods of unit length.

For each of the four choices of $(\mu(t),\varsigma_0)$, and each value of $h$ in an equispaced grid, the following procedure was repeated in 500 Monte Carlo trials. We generated i.i.d.~realizations $S_1(t),\dots,S_{n+m_n}(t)$ of the sample path defined in~\eqref{eqn:GBM} with $n+m_n=300$ and $m_n=30$.  For each of these sets of functional observations, we constructed simultaneous $(1-\alpha)$-confidence intervals $\hat{\mathcal{I}}_1,\dots,\hat{\mathcal{I}}_p$ for the parameters $\beta_1,\dots,\beta_p$, as described above, with $p=100$ and $\alpha=5\%$. Whenever any of these intervals excluded 0, a rejection was recorded, and the rejection rate among the 500 trials was plotted as a function of $h$ in Figures~\ref{fig:simu apple}-\ref{fig:simu moderna}. The corresponding rejection rate based on the simultaneous confidence intervals proposed in the paper~\citep{Fan:2023}  was also plotted in the same way. In all four figures, the nominal level of $\alpha=5\%$ is marked with a dashed horizontal line.

To illustrate the characteristics of the simulated functional data, ten realizations of $S(t)$ based on $h=0$ with $\varsigma_0$ corresponding to Apple stock are plotted in Figure~\ref{fig:gbm null}. In the same setting, Figure~\ref{fig:simu apple sd} displays estimates of the sorted values $\sigma_{(1)}\geq\cdots\geq \sigma_{(p)}$, where $\sigma_j^2$ is the variance of the $j$th Fourier coefficient of $S_1(t)$. In particular, Figure~\ref{fig:simu apple sd}
shows a clear variance decay profile.\\

\begin{figure}[ht]
  \centering
  \begin{minipage}{.47\textwidth}
    \centering
    \begin{overpic}[width=\textwidth]{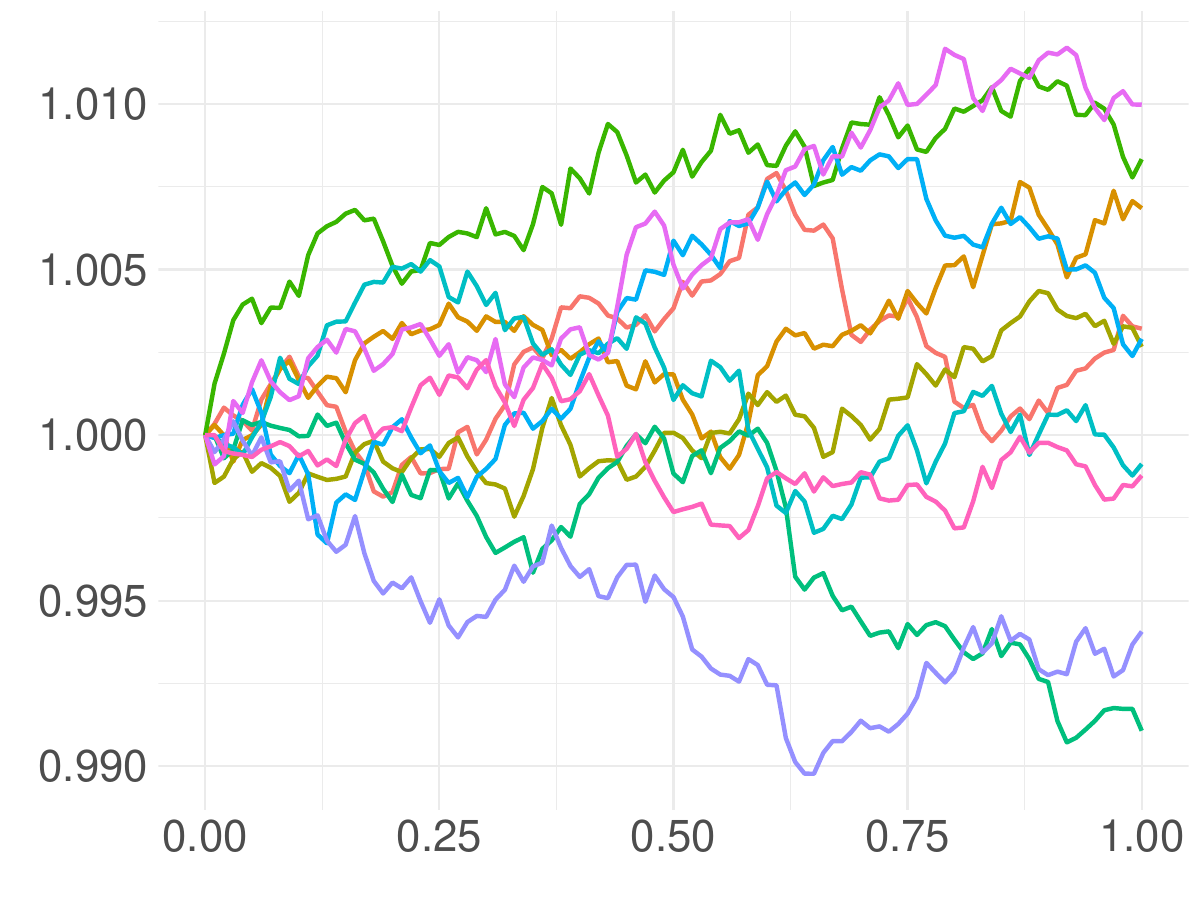}
\put(-5,35){\rotatebox{90}{$S(t)$}}
\put(53,-3){ $t$  }
	\end{overpic}
 \caption{Representative sample paths of $S(t)=\exp\big( (h \mu(t)-\varsigma_0^2/2)t+\varsigma_0 W(t)\big)$ when $h=0$, and $\varsigma_0$ is selected based on historical price data for Apple stock.}
 \label{fig:gbm null}
  \end{minipage}
  \hfill
  \begin{minipage}{.47\textwidth}
    \centering
    \begin{overpic}[width=\textwidth]{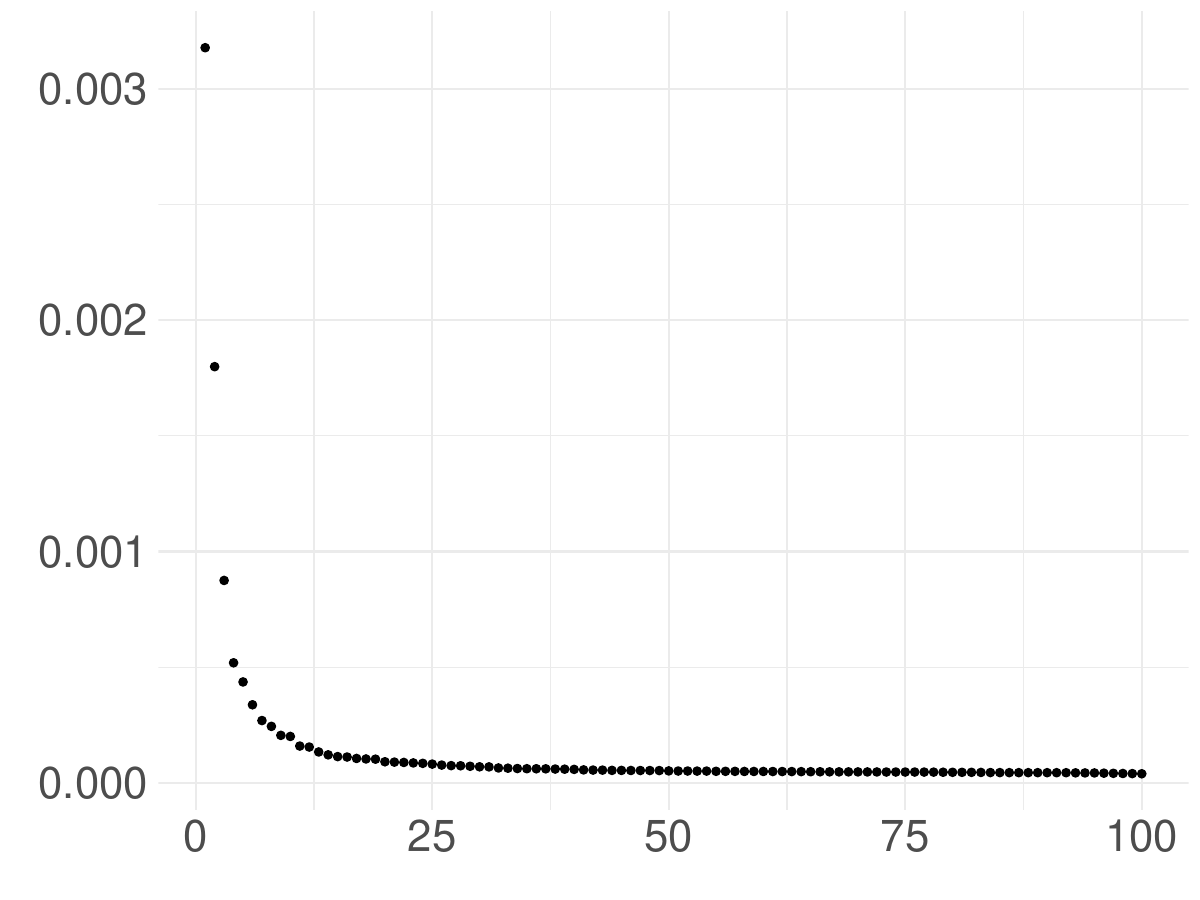}
\put(50,-1){ $j$  }
    \end{overpic}
    \caption{Estimates of the ordered values $\sigma_{(1)}\geq\cdots\geq\sigma_{(p)}$, where $\sigma_j^2$ denotes the variance of the $j$th Fourier coefficient of sample paths generated as in Figure~\ref{fig:gbm null}.}
    \label{fig:simu apple sd}
  \end{minipage}
\end{figure}

\noindent \textbf{Preparation of stock data.}
Price data for 4 stocks (Apple, Nvidia, Moderna, and JP Morgan) were collected from the Alpha Vantage database \citep{alphavantage2024} during every minute of trading between March 1, 2024 and March 22, 2024, including pre-market and after-hours trading.
The data were divided into 100-minute intervals (normalized to unit length), and each interval was divided into time points $t_0,t_1,\ldots,t_{100}$ spaced one minute apart.
For a given stock, letting $P(t_j)$ denote its price at time $t_j$, we computed the cumulative log return curve within the interval as $R(t_j) = \sum_{\ell=1}^j \log\big(P(t_\ell)/P(t_{\ell-1})\big) = \log\big(P(t_j)/P(t_0)\big)$, $j=1,\ldots,100$.
In this way, we obtained one discretely observed realization of the function $R(t)$ over each 100-minute interval. To promote independence among these functional observations, we only retained them from \emph{every other} 100-minute interval, ensuring that they are separated by gaps of 100 minutes. We also excluded functional observations that were obtained from intervals with missing data.\\

\begin{figure}[ht]
    \centering
    \begin{subfigure}{0.45\textwidth}
        \centering
        \begin{overpic}[width=\linewidth]{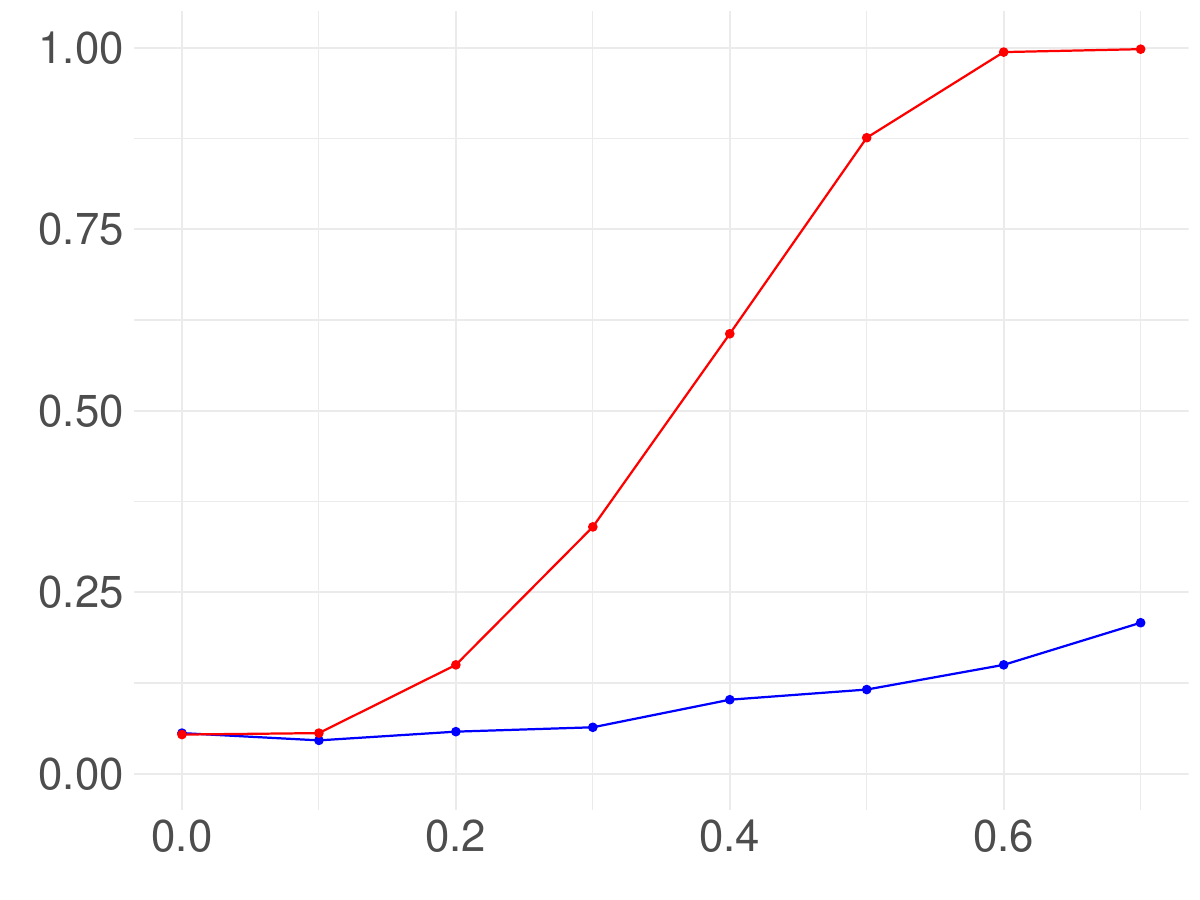}
            \put(-5,30){\rotatebox{90}{rejection rate}}
            \put(50,-1){ $h$ } 
            \put(20, 60){\color{red} \--\!\--\!\--\!\--\!\--\!\-- \small PM}
            \put(20, 55){\color{blue} \--\!\--\!\--\!\--\!\--\!\-- \small HL}
            \put(0, 13.2){\color{black}{\dashline[50]{1}(0,0)(100,0)}}
        \end{overpic}
        \caption{Apple stock}
        \label{fig:simu apple}
    \end{subfigure}
    \hfill
    \begin{subfigure}{0.45\textwidth}
        \centering
        \begin{overpic}[width=\linewidth]{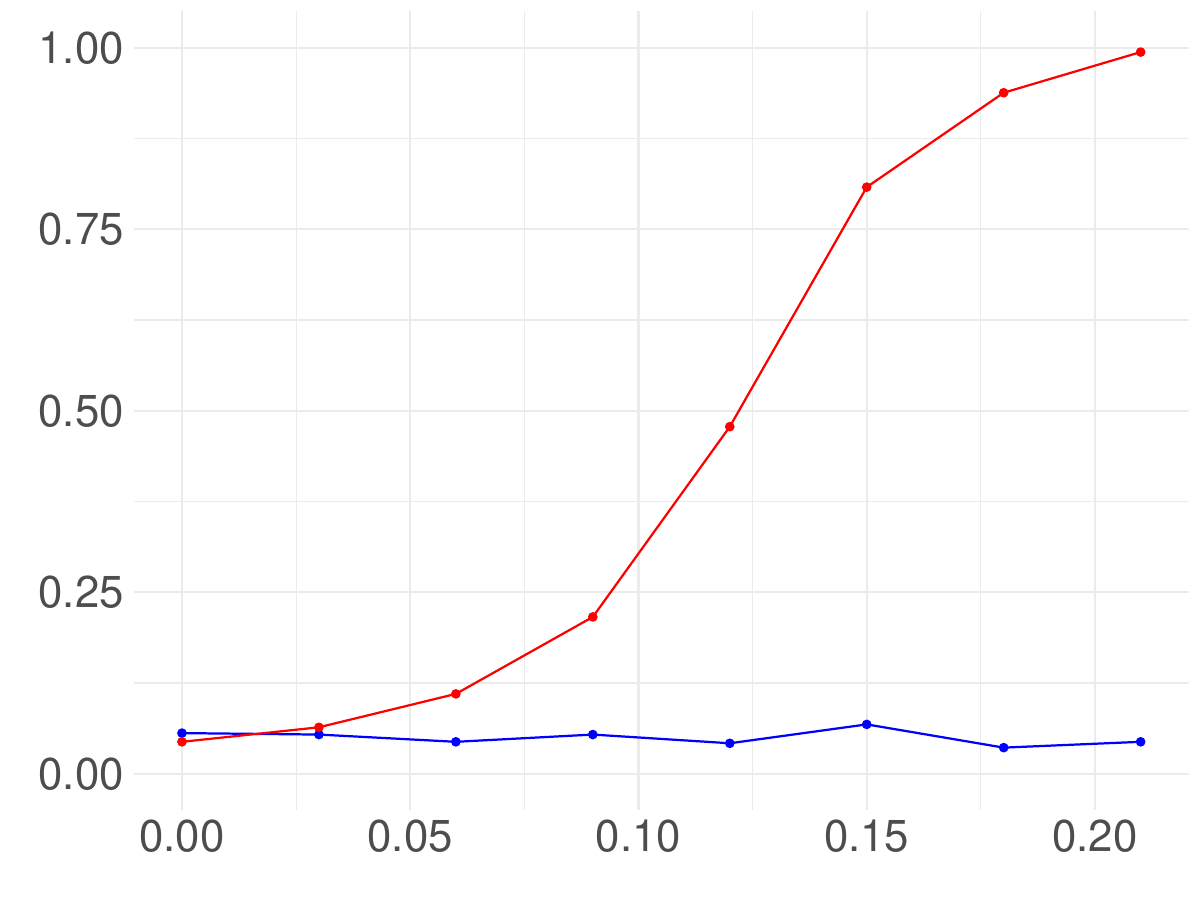}
            \put(-5,30){\rotatebox{90}{rejection rate}}
            \put(50,-1){ $h$ }  
            \put(20, 60){\color{red} \--\!\--\!\--\!\--\!\--\!\-- \small PM}
            \put(20, 55){\color{blue} \--\!\--\!\--\!\--\!\--\!\-- \small HL}
            \put(0, 13.2){\color{black}{\dashline[50]{1}(0,0)(100,0)}}
        \end{overpic}
        \caption{Nvidia stock}
        \label{fig:simu nvidia}
    \end{subfigure}

    \vskip\baselineskip

    \begin{subfigure}{0.45\textwidth}
        \centering
        \begin{overpic}[width=\linewidth]{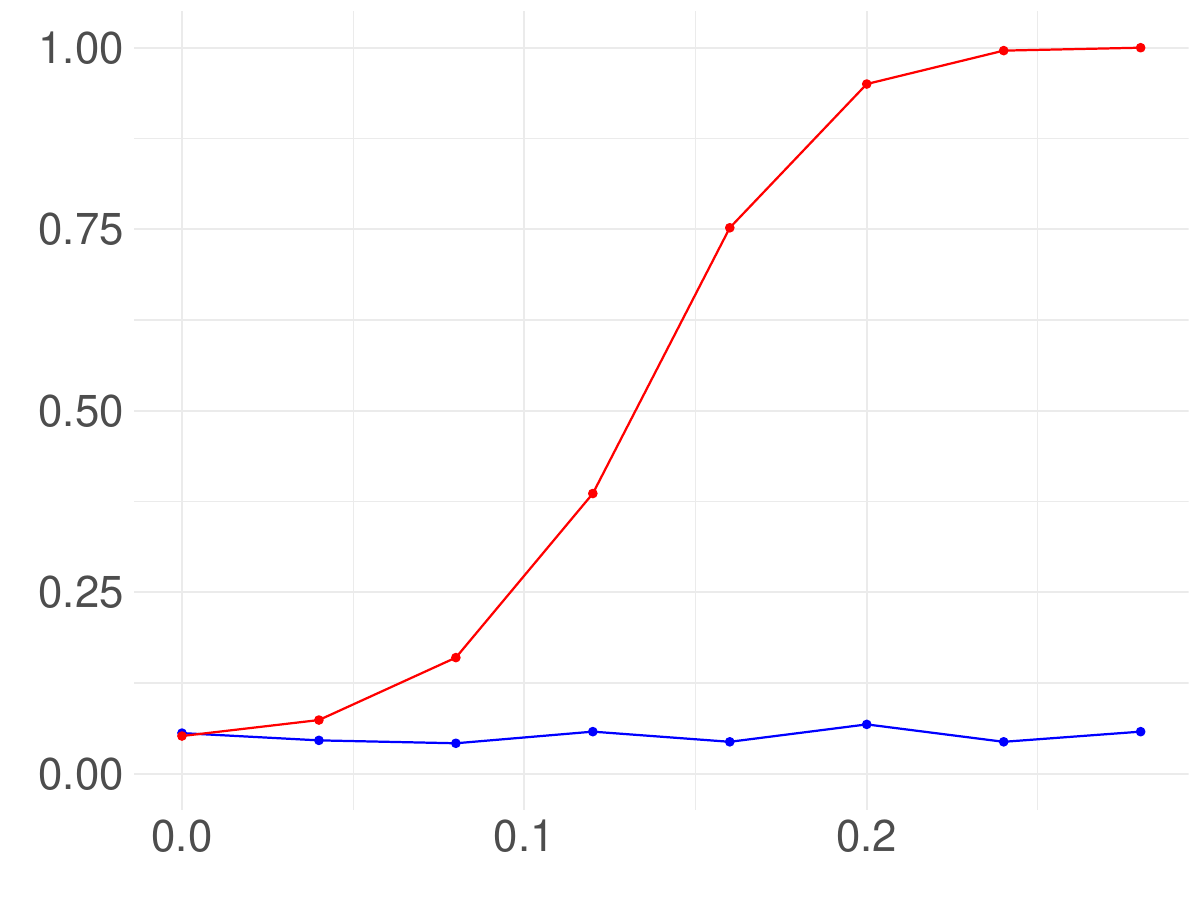}
            \put(-5,30){\rotatebox{90}{rejection rate}}
            \put(50,-1){ $h$ }   
            \put(20, 60){\color{red} \--\!\--\!\--\!\--\!\--\!\-- \small PM}
            \put(20, 55){\color{blue} \--\!\--\!\--\!\--\!\--\!\-- \small HL}
            \put(0, 13.2){\color{black}{\dashline[50]{1}(0,0)(100,0)}}
        \end{overpic}
        \caption{JP Morgan stock}
        \label{fig:simu jpm}
    \end{subfigure}
    \hfill
    \begin{subfigure}{0.45\textwidth}
        \centering
        \begin{overpic}[width=\linewidth]{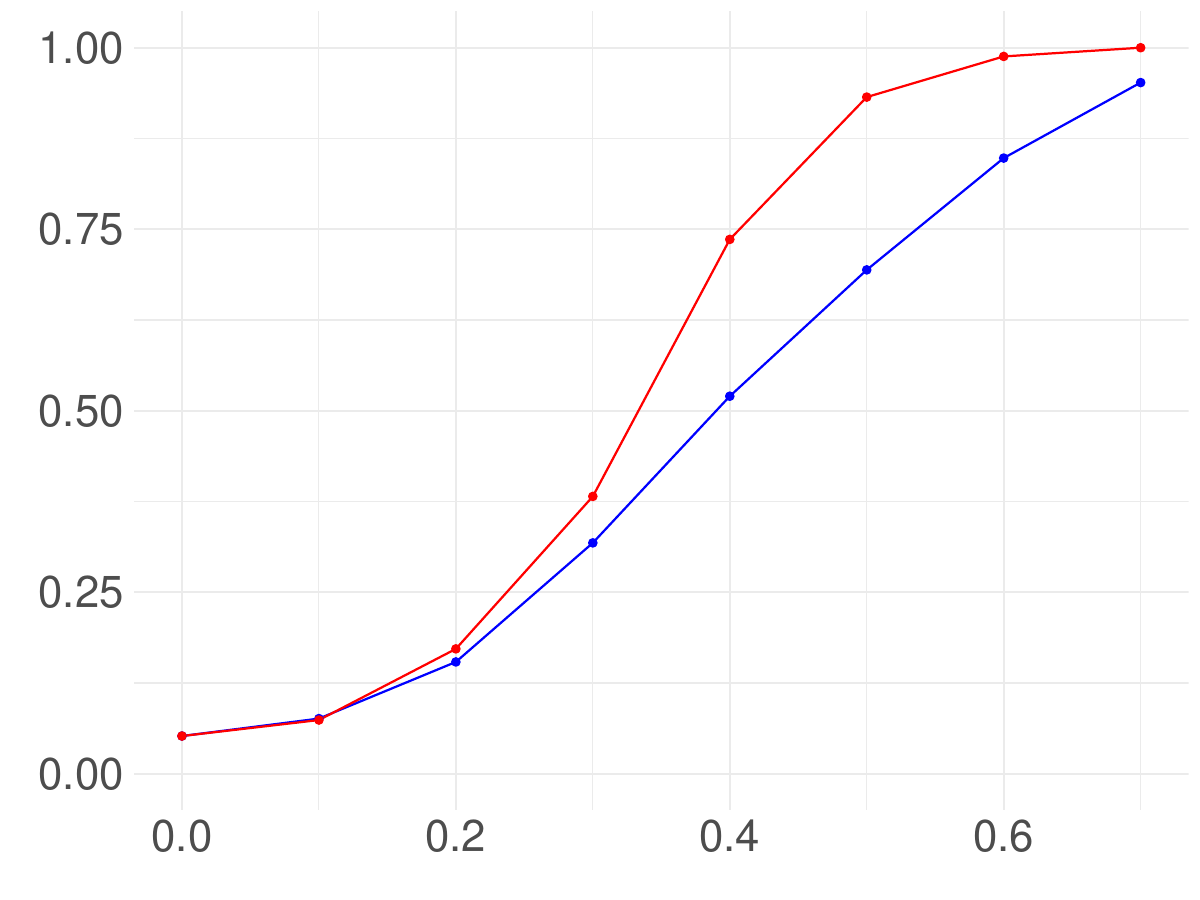}
            \put(-5,30){\rotatebox{90}{rejection rate}}
            \put(50,-1){ $h$ }   
            \put(20, 60){\color{red} \--\!\--\!\--\!\--\!\--\!\-- \small PM}
            \put(20, 55){\color{blue} \--\!\--\!\--\!\--\!\--\!\-- \small HL}
            \put(0, 13.2){\color{black}{\dashline[50]{1}(0,0)(100,0)}}
        \end{overpic}
        \caption{Moderna stock}
        \label{fig:simu moderna}
    \end{subfigure}

    \caption{The panels compare the rejection rates of the methods PM and HL when functional observations are generated in the form~\eqref{eqn:GBM} and the parameters $(\mu(t),\varsigma_0)$ are based on historical stock price data for Apple, Nvidia, JP Morgan, and Moderna.}
    \label{fig:simu_all}
\end{figure}
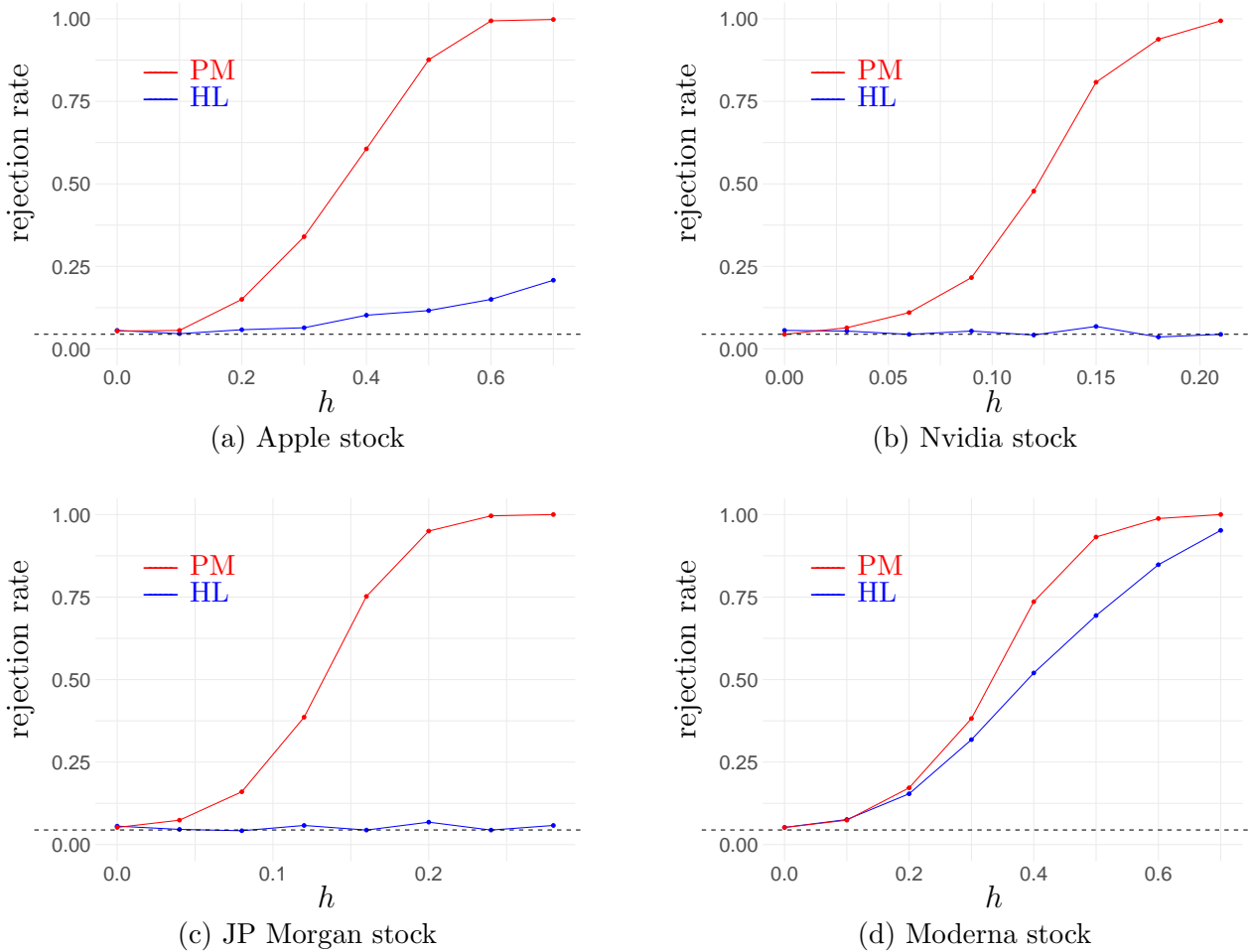

\noindent\textbf{Discussion of results.} In Figures~\ref{fig:simu apple}-\ref{fig:simu moderna}, the rejection rate curves for the proposed method are labeled by PM, and the corresponding curves based on the intervals proposed in~\citep{Fan:2023} are labeled HL. Recall that $h=0$ under the null hypothesis $\mathsf{H}_0$, and so the value of a curve at $h=0$ represents the empirical level. It is clear that in all four panels of Figure~\ref{fig:simu_all}, the empirical levels of both methods closely match the nominal level of $\alpha=5\%$, marked with a dashed horizontal line.
 
 However, the methods differ markedly in terms of power, which is represented by the values of the curves at $h>0$. In the three settings based on the stock data of Apple, Nvidia, and JP Morgan, the power of PM increases steadily with $h$, whereas the power of HL stays relatively flat. In the setting based on Moderna stock data, the two methods are more competitive, but even here, the power of PM is still noticeably higher for most values of $h$. The power advantage of PM is understandable in light of the characteristics of the simultaneous confidence intervals $\hat{\mathcal{I}}_1,\dots,\hat{\mathcal{I}}_p$ for $\beta_1,\dots,\beta_p$. (Recall that both methods reject the null hypothesis $\mathsf{H}_0: h=0$ whenever any of their associated intervals exclude 0.) As was observed in Section~\ref{sec: euclidean simu}, the intervals produced by PM tend to be tighter than those produced by HL, and tighter intervals make it easier to exclude 0, resulting in higher power.

\section*{Acknowledgement}
We are grateful to Mengxin Yu for generously providing the code for the HL method.

\bibliographystyle{abbrv} 
\bibliography{library}

@book{Oksendal,
  title={Stochastic Differential Equations: An Introduction with Applications},
  author={Oksendal, Bernt},
  year={2013},
  publisher={Springer}
}

@book{johnstone,
title={Gaussian Estimation:
Sequence and Wavelet Models},
author={Johnstone, I. M.},
year={2019+},
publisher={preprint},
url={https://imjohnstone.su.domains/GE_09_16_19.pdf}
}

@book{muirhead2009aspects,
  title={Aspects of Multivariate Statistical Theory},
  author={Muirhead, Robb J},
  year={2009},
  publisher={John Wiley \& Sons}
}

@book{kotz2019continuous,
  title={Continuous Multivariate Distributions, Volume 1: Models and Applications},
  author={Kotz, Samuel and Balakrishnan, Narayanaswamy and Johnson, Norman L},
  volume={334},
  year={2019},
  publisher={John Wiley \& Sons}
}

@book{Bai:2010,
  title={Spectral Analysis of Large Dimensional Random Matrices},
  author={Bai, Z. and Silverstein, J. W.},
  year={2010},
  publisher={Springer}
}

@book{Comon:2010,
  title={Handbook of Blind Source Separation: Independent Component Analysis and Applications},
  author={Comon, P. and Jutten, C.},
  year={2010},
  publisher={Academic Press}
}

@article{rudelson2015small,
  title={Small ball probabilities for linear images of high-dimensional distributions},
  author={Rudelson, Mark and Vershynin, Roman},
  journal={International Mathematics Research Notices},
  volume={2015},
  number={19},
  pages={9594--9617},
  year={2015}
}

@book{Foss,
  title={An Introduction to Heavy-tailed and Subexponential Distributions},
  author={Foss, Sergey and Korshunov, Dmitry and Zachary, Stan},
  year={2011},
  publisher={Springer}
}

@book{Nair,
  title={The Fundamentals of Heavy Tails: Properties, Emergence, and Estimation},
  author={Nair, Jayakrishnan and Wierman, Adam and Zwart, Bert},
  year={2022},
  publisher={Cambridge}
}

@article{Rosenthal,
  title={Best constants in moment inequalities for linear combinations of independent and exchangeable random variables},
  author={Johnson, William B and Schechtman, Gideon and Zinn, Joel},
  journal={The Annals of Probability},
  pages={234--253},
  year={1985}
}

@book{stein1971introduction,
  title={Introduction to {F}ourier Analysis on Euclidean Spaces},
  author={Stein, Elias M and Weiss, Guido},
  year={1971},
  publisher={Princeton University Press}
}

@article{Lou:2017,
  title={Simultaneous Inference for High Dimensional Mean Vectors},
  author={Lou, Zhipeng and Wu, Wei Biao},
  journal={arXiv:1704.04806},
  year={2017}
}

@article{ke2019user,
  title={User-friendly covariance estimation for heavy-tailed distributions},
  author={Ke, Yuan and Minsker, Stanislav and Ren, Zhao and Sun, Qiang and Zhou, Wen-Xin},
  journal={Statistical Science},
  volume={34},
  number={3},
  pages={454--471},
  year={2019},
  publisher={JSTOR}
}

@article{lugosi2019mean,
  title={Mean estimation and regression under heavy-tailed distributions: A survey},
  author={Lugosi, G{\'a}bor and Mendelson, Shahar},
  journal={Foundations of Computational Mathematics},
  volume={19},
  number={5},
  pages={1145--1190},
  year={2019},
  publisher={Springer}
}

@article{lopes2022sharp,
  title={A sharp lower-tail bound for {G}aussian maxima with application to bootstrap methods in high dimensions},
  author={Lopes, Miles E and Yao, Junwen},
  journal={Electronic Journal of Statistics},
  volume={16},
  number={1},
  pages={58--83},
  year={2022},
  publisher={The Institute of Mathematical Statistics and the Bernoulli Society}
}

@article{Bentkus:2003,
  title={On the dependence of the {B}erry-{E}sseen bound on dimension},
  author={Bentkus, Vidmantas},
  journal={Journal of Statistical Planning and Inference},
  volume={113},
  number={2},
  pages={385--402},
  year={2003},
  publisher={Elsevier}
}

@article{CCK:2013,
author = {Victor Chernozhukov and Denis Chetverikov and Kengo Kato},
title = {{Gaussian approximations and multiplier bootstrap for maxima of sums of high-dimensional random vectors}},
volume = {41},
journal = {The Annals of Statistics},
number = {6},
publisher = {Institute of Mathematical Statistics},
pages = {2786--2819},
year = {2013},
doi = {10.1214/13-AOS1161},
URL = {https://doi.org/10.1214/13-AOS1161}
}

@article{lopes2020bootstrapping,
author = {Miles E. Lopes and Zhenhua Lin and Hans-Georg M{\"u}ller},
title = {{Bootstrapping max statistics in high dimensions: Near-parametric rates under weak variance decay and application to functional and multinomial data}},
volume = {48},
journal = {The Annals of Statistics},
number = {2},
publisher = {Institute of Mathematical Statistics},
pages = {1214--1229},
year = {2020},
doi = {10.1214/19-AOS1844},
URL = {https://doi.org/10.1214/19-AOS1844}
}

@article{vershynin2012introduction,
  title={Introduction to the non-asymptotic analysis of random matrices},
  author={Vershynin, Roman},
  journal={Compressed Sensing},
  pages={210--268},
  year={2012},
  publisher={Cambridge University Press}
}

@article{chernozhukov2016empirical,
  title={Empirical and multiplier bootstraps for suprema of empirical processes of increasing complexity, and related {G}aussian couplings},
  author={Chernozhukov, Victor and Chetverikov, Denis and Kato, Kengo},
  journal={Stochastic Processes and their Applications},
  volume={126},
  number={12},
  pages={3632--3651},
  year={2016},
  publisher={Elsevier}
}

@inproceedings{nazarov2003maximal,
  title={On the Maximal Perimeter of a Convex Set in $\mathbb{R}^n$ with Respect to a {G}aussian Measure},
  author={Nazarov, Fedor},
  booktitle={Geometric Aspects of Functional Analysis: Israel Seminar 2001-2002},
  pages={169--187},
  year={2003}
}

@book{van1996weak,
  title={Weak Convergence and Empirical Processes},
  author={van der Vaart, Aad and Wellner, Jon A.},
  publisher={Springer},
  year={1996}
}

@article{rio2017constants,
author = {Emmanuel Rio},
title = {{About the constants in the {F}uk-{N}agaev inequalities}},
volume = {22},
journal = {Electronic Communications in Probability},
publisher = {Institute of Mathematical Statistics and Bernoulli Society},
pages = {1--12},
year = {2017},
doi = {10.1214/17-ECP57},
URL = {https://doi.org/10.1214/17-ECP57}
}

@article{Giessing:processes,
  title={Gaussian and Bootstrap Approximations for Suprema of Empirical Processes},
  author={Giessing, Alexander},
  journal={arXiv:2309.01307},
  year={2023}
}

@article{Chen:ridges,
author = {Yen-Chi Chen and Christopher R. Genovese and Larry Wasserman},
title = {{Asymptotic theory for density ridges}},
volume = {43},
journal = {The Annals of Statistics},
number = {5},
publisher = {Institute of Mathematical Statistics},
pages = {1896--1928},
year = {2015}
}

@article{Chen:Jackknife,
  title={Jackknife multiplier bootstrap: finite sample approximations to the {$U$}-process supremum with applications},
  author={Chen, Xiaohui and Kato, Kengo},
  journal={Probability Theory and Related Fields},
  volume={176},
  pages={1097--1163},
  year={2020},
  publisher={Springer}
}

@article{Han:Bernoulli,
author = {Fang Han and Sheng Xu and Wen-Xin Zhou},
title = {{On Gaussian comparison inequality and its application to spectral analysis of large random matrices}},
volume = {24},
journal = {Bernoulli},
number = {3},
pages = {1787--1833},
year = {2018}
}

@article{Dette:Banach,
  title={Functional data analysis in the {B}anach space of continuous functions},
  author={Dette, Holger and Kokot, Kevin and Aue, Alexander},
  journal={The Annals of Statistics},
  volume={48},
  number={2},
  pages={1168--1192},
  year={2020}
}

@article{Lopes:2022:AOS,
  title={Central limit theorem and bootstrap approximation in high dimensions: Near $1/\sqrt{n}$ rates via implicit smoothing},
  author={Lopes, Miles E},
  journal={The Annals of Statistics},
  volume={50},
  number={5},
  pages={2492--2513},
  year={2022},
  publisher={Institute of Mathematical Statistics}
}

@article{spokoiny2015bootstrap,
 ISSN = {00905364},
 URL = {http://www.jstor.org/stable/43818864},
 author = {Vladimir Spokoiny and Mayya Zhilova},
 journal = {The Annals of Statistics},
 number = {6},
 pages = {2653--2675},
 publisher = {Institute of Mathematical Statistics},
 title = {BOOTSTRAP CONFIDENCE SETS UNDER MODEL MISSPECIFICATION},
 urldate = {2024-05-28},
 volume = {43},
 year = {2015}
}

@book{ruppertstatistics,
  title={Statistics and {D}ata {A}nalysis for {F}inancial {E}ngineering with {R} {E}xamples},
  author={Ruppert, David and Matteson, David S},
  year={2011},
  publisher={Springer}
}

@article{Sun:2022,
  title={An omnibus test for detection of subgroup treatment effects via data partitioning},
  author={Sun, Y. and He, X. and Hu, J.},
  journal={Annals of Applied Statistics},
  volume={16},
  number={4},
  pages={2266--2278},
  year={2022}
}

@article{CCK:2017,
 ISSN = {00911798},
 URL = {http://www.jstor.org/stable/26362255},
 author = {Victor Chernozhukov and Denis Chetverikov and Kengo Kato},
 journal = {The Annals of Probability},
 number = {4},
 pages = {2309--2352},
 publisher = {Institute of Mathematical Statistics},
 title = {CENTRAL LIMIT THEOREMS AND BOOTSTRAP IN HIGH DIMENSIONS},
 urldate = {2024-05-29},
 volume = {45},
 year = {2017}
}

@article{Giessing:2020,
  title={Bootstrapping $\ell_p$-Statistics in High Dimensions},
  author={Giessing, Alexander and Fan, Jianqing},
  journal={arXiv:2006.13099},
  year={2020}
}

@article{Deng:2020,
  title={Beyond {G}aussian approximation: Bootstrap for maxima of sums of independent random vectors},
  author={Deng, Hang and Zhang, Cun-Hui},
  journal={The Annals of Statistics},
  volume={48},
  number={6},
  pages={3643--3671},
  year={2020}
}

@article{Fang:2023,
  title={High-dimensional Central Limit Theorems by {S}tein's Method in the Degenerate Case},
  author={Fang, Xiao and Koike, Yuta and Liu, Song-Hao and Zhao, Yi-Kun},
  journal={arXiv:2305.17365},
  year={2023}
}

@article{Yu:2021,
  title={Finite sample change point inference and identification for high-dimensional mean vectors},
  author={Yu, M. and Chen, X.},
  journal={Journal of the Royal Statistical Society Series B: Statistical Methodology},
  volume={83},
  number={2},
  pages={247--270},
  year={2021}
}

@article{Koike:2024,
  title={High-dimensional bootstrap and asymptotic expansion},
  author={Koike, Yuta},
  journal={arXiv:2404.05006},
  year={2024}
}

@article{CCK:2023,
  title={Nearly optimal central limit theorem and bootstrap approximations in high dimensions},
  author={Chernozhukov, Victor and Chetverikov, Denis and Koike, Yuta},
  journal={The Annals of Applied Probability},
  volume={33},
  number={3},
  pages={2374--2425},
  year={2023},
  publisher={Institute of Mathematical Statistics}
}

@article{Preinerstorfer:2024,
title = {A remark on moment-dependent phase transitions in high-dimensional {G}aussian approximations},
journal = {Statistics $\&$ Probability Letters},
volume = {211},
pages = {110149},
year = {2024},
issn = {0167-7152},
doi = {https://doi.org/10.1016/j.spl.2024.110149},
author = {Anders Bredahl Kock and David Preinerstorfer}
}

@article{CCK_bands,
author = {Victor Chernozhukov and Denis Chetverikov and Kengo Kato},
title = {{Anti-concentration and honest, adaptive confidence bands}},
volume = {42},
journal = {The Annals of Statistics},
number = {5},
publisher = {Institute of Mathematical Statistics},
pages = {1787--1818},
year = {2014}
}

@article{zhang2017gaussian,
author = {Danna Zhang and Wei Biao Wu},
title = {{Gaussian approximation for high dimensional time series}},
volume = {45},
journal = {The Annals of Statistics},
number = {5},
publisher = {Institute of Mathematical Statistics},
pages = {1895--1919},
year = {2017},
doi = {10.1214/16-AOS1512},
URL = {https://doi.org/10.1214/16-AOS1512}
}

@article{mendelson2020robust,
author = {Shahar Mendelson and Nikita Zhivotovskiy},
title = {{Robust covariance estimation under $L_{4}$-$L_{2}$ norm equivalence}},
volume = {48},
journal = {The Annals of Statistics},
number = {3},
publisher = {Institute of Mathematical Statistics},
pages = {1648--1664},
year = {2020},
doi = {10.1214/19-AOS1862},
URL = {https://doi.org/10.1214/19-AOS1862}
}

@article{roy2021empirical,
  title={On empirical risk minimization with dependent and heavy-tailed data},
  author={Roy, Abhishek and Balasubramanian, Krishnakumar and Erdogdu, Murat A},
  journal={Advances in Neural Information Processing Systems},
  volume={34},
  pages={8913--8926},
  year={2021}
}

@article{Fan:2023,
  title={Robust high-dimensional tuning free multiple testing},
  author={Fan, Jianqing and Lou, Zhipeng and Yu, Mengxin},
  journal={The Annals of Statistics},
  volume={51},
  number={5},
  pages={2093--2115},
  year={2023},
  publisher={Institute of Mathematical Statistics}
}

@inproceedings{polya1949remarks,
  title={Remarks on characteristic functions},
  author={P{\'o}lya, George},
  booktitle={Proceedings of the First Berkeley Symposium on Mathematical Statistics and Probability},
  pages={115--123},
  year={1949}
}

@article{abdalla2022covariance,
  title={Covariance estimation: Optimal dimension-free guarantees for adversarial corruption and heavy tails},
  author={Abdalla, Pedro and Zhivotovskiy, Nikita},
  journal={arXiv:2205.08494},
  year={2022}
}

@article{Hodges:1963,
  title={Estimates of Location Based on Rank Tests},
  author={Hodges, JL and Lehmann, EL},
  journal={The Annals of Mathematical Statistics},
  volume={34},
  number={2},
  pages={598--611},
  year={1963},
  publisher={Institute of Mathematical Statistics}
}

@article{Kuchibhotla:2020:PSI,
  title={Valid post-selection inference in model-free linear regression},
  author={Kuchibhotla, Arun K and Brown, Lawrence D and Buja, Andreas and Cai, Junhui and George, Edward I and Zhao, Linda H},
  journal={The Annals of Statistics},
  volume={48},
  number={5},
  pages={2953--2981},
  year={2020},
  publisher={JSTOR}
}

@article{Kuchibhota_2021,
author = {Arun Kumar Kuchibhotla and Somabha Mukherjee and Debapratim Banerjee},
title = {{High-dimensional {CLT}: Improvements, non-uniform extensions and large deviations}},
volume = {27},
journal = {Bernoulli},
number = {1},
pages = {192 -- 217},
year = {2021}
}

@article{Kuchibhotla2020high,
  title={High-dimensional {CLT} for Sums of Non-degenerate Random Vectors: $n^{-1/2}$-rate},
  author={Kuchibhotla, Arun Kumar and Rinaldo, Alessandro},
  journal={arXiv:2009.13673},
  year={2020}
}

@article{Singh:2023,
  title={Kernel Ridge Regression Inference},
  author={Singh, Rahul and Vijaykumar, Suhas},
  journal={arXiv:2302.06578},
  year={2023}
}

@article{Chetverikov:shape,
  title={The econometrics of shape restrictions},
  author={Chetverikov, Denis and Santos, Andres and Shaikh, Azeem M},
  journal={Annual Review of Economics},
  volume={10},
  pages={31--63},
  year={2018},
  publisher={Annual Reviews}
}

@article{Lopes:Bernoulli,
  title={Bootstrapping the operator norm in high dimensions: Error estimation for covariance matrices and sketching},
  author={Lopes, Miles E and Erichson, N Benjamin and Mahoney, Michael W},
  journal={Bernoulli},
  volume={29},
  number={1},
  pages={428--450},
  year={2023},
  publisher={Bernoulli Society for Mathematical Statistics and Probability}
}

@misc{alphavantage2024,
  key = {Alpha Vantage},
  title = {Alpha {V}antage {API}},
  howpublished = {https://www.alphavantage.co},
  year = {2024}
}

@book{Nemirovsky1983,
  title={Problem Complexity and Method Efficiency in Optimization},
  author={Nemirovsky, A. S. and Yudin, D. B.},
  publisher={Wiley},
  year={1983}
}

\newpage

\appendix

\section*{\centering Supplementary Material\\ Robust Max Statistics for High-Dimensional Inference}

Appendix~\ref{sec:prelimsupp} introduces preliminary material not covered in the main text. Appendix~\ref{supp:proofthm1} outlines the proof of Theorem~\ref{thm:main}, and the main supporting arguments are given in Appendices~\ref{sec:Gaussian}-\ref{sec:boot}. Appendix~\ref{sec:mom} contains technical results on median-of-means estimators. Appendix~\ref{supp: moments eg} proves Proposition~\ref{prop:moments}. Lastly, Appendix~\ref{sec:background} contains various background results.

\section{Preliminaries for supplementary material}\label{sec:prelimsupp}

\noindent \textbf{Notation.}  The distribution of any random variable $U$ is denoted as $\mathcal{L}(U)$. We write $\mathcal{L}(U|X)$ to refer to the conditional distribution of $U$ given both the hold-out and non-hold-out sets of observations, whereas we write $\mathcal{L}(U|X')$ to refer to the conditional distribution of $U$ given the hold-out set only. Similarly, we use $\P(\cdot|X)$ and $\P(\cdot | X')$ to denote conditional probabilities in the two cases just mentioned, and we use $\|\cdot\|_{L^q|X}$ and $\|\cdot\|_{L^q|X'}$ to denote the corresponding conditional $L^q$ norms.

For any $d\in \{1,\ldots,p\}$, recall that $J(d)$ denotes a set of indices corresponding to the $d$ largest values among $\{\sigma_1, \ldots, \sigma_p\}$. That is, $\{\sigma_{(1)}, \ldots, \sigma_{(d)} \} = \{ \sigma_j | j \in J(d) \}$. Letting $l_n$ be as in Assumption~\ref{A:model}, define the integer
$$k_n = l_n^5\wedge p,$$ 
which always satisfies $1\leq l_n\leq k_n\leq p$. For each $d\in\{1,\dots,p\}$, let
\begin{equation}
    M_{d}(X) = \max_{j\in J(d)}\frac{1}{\sigma_{j}^{\tau} \sqrt n}\sum_{i=1}^n (X_{ij}-\mu_j).
\end{equation}
Letting $G(X)=(G_1(X), \ldots, G_p(X))$ be a centered Gaussian random vector with the same covariance matrix as $X_1$, the Gaussian counterpart of $M_d(X)$ is defined as
\begin{align*}
 \tilde M_{d}(X) = \max_{j \in J(d)} \frac{G_j(X)}{\sigma_j^{\tau}}.
 \end{align*}
Next, for each $i\in\{1,\dots,n\}$ and $j,d\in\{1,\dots,p\}$, define
\begin{equation*}
 \ \ \ \ \ \ \ \ \ \ \ \ \ \  Y_{ij} = \varphi_{t_j}(X_{ij}-\mu_j) \text{ \ \ \ and \ \ \ } M_{d}(Y)=\max_{j\in J(d)}\frac{1}{\sigma_{j}^{\tau} \sqrt n}\sum_{i=1}^n Y_{ij}-\E(Y_{ij}),
\end{equation*}
as well as
\begin{equation*}
    \hat Y_{ij} = \varphi_{\hat t_j}(X_{ij}-\mu_j)  \text{ \ \ \ and \ \ \  }  M_d(\hat Y)=\max_{j\in J(d)}\frac{1}{\hat \sigma_{j}^{\tau} \sqrt n}\sum_{i=1}^n \hat Y_{ij}.
\end{equation*}
Let $\xi_1,\dots,\xi_n$ be i.i.d.~standard Gaussian random variables, generated independently of $X_1, \ldots, X_{n+m_n}$, and define
\begin{align}
\label{equa: MdstarX}
 M_d^*(X)  =\max_{j\in J(d)}\frac{1}{ \sigma_{j}^{\tau} \sqrt n}\sum_{i=1}^n \xi_i( X_{ij}- \bar X_j).
\end{align}
Let $\tilde X_j$ denote the median-of-means estimator for $\mu_j$ described in Section \ref{sec:method} and define 
\begin{align*}
    \hat Z_{ij}  &= \varphi_{\hat t_j} (X_{ij}-\tilde X_j) \text{ \ \ \ \ and  \ \ \  }
 M_d^*(\hat Z)=\max_{j\in J(d)}\frac{1}{\hat \sigma_{j}^{\tau} \sqrt n}\sum_{i=1}^n \xi_i(\hat Z_{ij}-\hat Z_j).
\end{align*}
where we let $\hat Z_j =\frac{1}{n}\sum_{i=1}^n \hat Z_{ij}$.\\

\noindent \textbf{Frequently-used inequalities.}  As a shorthand for the Kolmogorov metric between generic random variables $U$ and $V$ we write
$$
d_{\mathrm{K}}(\mathcal{L}(U), \mathcal{L}(V)) = \sup_{t \in \R} |\P(U \leq t ) - \P(V \leq t)|.
$$
We will often use the following two basic inequalities that hold for any random variables $U$ and $V$, and any number $s>0$,
\begin{align}\label{eqn:anticouple} 
d_{\mathrm{K}}(\mathcal{L}(U), \mathcal{L}(V)) \ \leq \  \sup _{t \in \mathbb{R}} \P(|V-t| \leq s) \, + \, \P(|U-V|>s),
\end{align}
and
\begin{equation}\label{eqn:anticonc_decomp}
        \sup_{t\in\R}\P(|U-t|\leq s) \ \leq \  \sup_{t\in\R}\P(|V-t|\leq s) + 2\dK(\L(U),\L(V)).
\end{equation}
If $g$ is a centered Gaussian random variable and $q \geq 1$, then there is an absolute constant $c>0$ such that
\begin{equation}\label{eqn:gaussianLq}
    \|g\|_{L^q} \ \leq \ c\sqrt q  \|g\|_{L^2},
\end{equation}
as recorded in~\citep[][Eqn.~2.11]{vershynin2012introduction}.
When referring to Chebyshev's inequality, we will typically use it in the following form for a generic random variable $U$,
\begin{equation}\label{eqn:Chebyshev}
    \P(|U|\geq e\|U\|_{L^q})\leq e^{-q}.
\end{equation}
For any random variables $U_1,\dots,U_p$, we have
\begin{align}
\label{eqn:maxout}
\Big \| \max_{1\leq j\leq p} U_j \Big \|_{L^q} \, \leq \, p^{1/q} \max_{1\leq j\leq p} \| U_j \|_{L^q}.    
\end{align}
Lastly, for any two real vectors $(a_1, \ldots, a_p)$ and $(b_1, \ldots, b_p)$, we have
\begin{align}
\label{eqn:maxdiff}
\Big|\max_{1 \leq j \leq p} a_j - \max_{1 \leq j \leq p} b_j \Big| \leq \max_{1 \leq j \leq p} |a_j - b_j |. 
\end{align}

\section{Proof of Theorem~\ref{thm:main}}\label{supp:proofthm1}

Observe that the left side of the bound in Theorem~\ref{thm:main} is given by
\begin{equation}
     \sup_{s\in\R}\Big| \P(\mathcal{M}_n\leq s) - \P(\mathcal{M}_n^*\leq s|X)\Big| = \dK\big(\L(M_p(\hat Y)), \L(M_p^*(\hat Z)|X)\big).
\end{equation}
We will bound this distance in three main parts 
\begin{align}
\dK\big(\L(M_p(\hat Y)), \L(M_p^*(\hat Z)|X)\big) & \ \  \leq \ \ 
\dK\big(\L(M_p(\hat Y)),\L(\tilde M_{k_n}(X))\big) \label{equ: piece1}\\
& \ \ \ \ \ \ + \ \dK\Big(\L(\tilde M_{k_n}(X)),\L(M_{k_n}^*(X)|X)\big) \label{equ: piece2}\\
& \ \ \ \ \ \ + \ \dK\big(\L( M_{k_n}^*(X)|X), \L(M_{p}^*(\hat Z)|X)\big). \label{equ: piece3}
\end{align}
The three terms on the right side respectively correspond to a Gaussian approximation, a Gaussian comparison, and a bootstrap approximation. These terms are respectively addressed in Proposition \ref{prop: piece2} of Appendix~\ref{sec:Gaussian}, Proposition~\ref{prop: gaus-compa} of Appendix~\ref{sec:comp} and Proposition~\ref{prop: XZ kn-p} of Appendix~\ref{sec:boot}, which show that all the terms are $\mathcal{O}(n^{-\frac{1}{2}+\epsilon})$ with probability at least $1-\mathcal{O}(n^{-\delta/4})$. \qed \\

\noindent\textbf{Conventions.} In the appendices supporting the proof of Theorem~\ref{thm:main}, we may assume without loss of generality that $\epsilon$ satisfies $\epsilon<1/2$ and that $n\geq c$ for any fixed constant $c>0$, for otherwise the result is trivially true. Also, we will often use $c$ to denote a generic positive constant not depending on $n$, whose value may differ at each appearance. Lastly, we may assume without loss of generality that $(\mu_1,\dots,\mu_p)=\E(X_1)=0$, because the conditions in Assumption~\ref{A:model} are shift invariant, and that $\max_{1\leq j\leq p}\sigma_j^2=1$, because the Kolmogorov metric is scale invariant.\label{sigmaconvention} To avoid repetitiveness, these conventions will not be stated explicitly in most of the results presented in the appendices.

\section{Gaussian approximation}\label{sec:Gaussian}

\begin{proposition}
\label{prop: piece2}
If the conditions of Theorem~\ref{thm:main} hold, then
\begin{equation*}
        \dK(\L(M_p(\hat Y)),\L(\tilde M_{k_n}(X))) \ \lesssim n^{-\frac{1}{2}+\epsilon}.
\end{equation*}   
\end{proposition}

\noindent \emph{Proof.}
The proof is based on the decomposition 
\begin{align*}
\dK(\L(M_p(\hat Y)),\L(\tilde M_{k_n}(X))) \leq \mathrm{I}_n + \mathrm{II}_n + \mathrm{III}_n,
\end{align*}
where we define
\begin{align*}
& \mathrm{I}_n = \dK(\L(M_p(\hat Y)), \L(M_{k_n}(\hat Y))), \\
& \mathrm{II}_n = \dK(\L(M_{k_n}(\hat Y)), \L(M_{k_n}(Y))),\\
& \mathrm{III}_n = \dK(\L(M_{k_n}(Y)), \L(\tilde M_{k_n}(X))).
\end{align*}
Below, the terms $\mathrm{I}_n$, $\mathrm{II}_n$, and $\mathrm{III}_n$ are shown to be at most of order $n^{-\frac{1}{2}+\epsilon}$  in Propositions \ref{prop: hatYp-kn}, \ref{Prop: practical kn}, and \ref{prop: Ykn-tildeXkn} respectively.

\begin{proposition}
\label{prop: hatYp-kn}
If the conditions of Theorem~\ref{thm:main} hold, then
    \begin{equation*}
        \dK(\L(M_p(\hat Y)),\L(M_{k_n}(\hat Y))) \ \lesssim n^{-\frac{1}{2}+\epsilon}.
    \end{equation*}
\end{proposition}

\noindent\emph{Proof.} 
For any $t \in \R$, define the events
\begin{align*}
A(t)=\Big\{\max_{j\in J(k_n)}\frac{\sum_{i=1}^n \hat Y_{ij}}{\sqrt{n}\hat\sigma_j^{\tau}} \leq t\Big\} \quad \text { and } \quad B(t)=\Big\{\max_{j\in J(k_n)^c}\frac{\sum_{i=1}^n \hat Y_{ij}}{\sqrt{n}\hat\sigma_j^{\tau}} > t\Big\}.
\end{align*}
It is straightforward to check that for any $t\in\R$, we have
\begin{equation*}
    \Big|\P(M_p(\hat Y)\leq t)-\P(M_{k_n}(\hat Y)\leq t)\Big| \ = \ \P(A(t)\cap B(t)).
\end{equation*}
Next, it can be checked that for any real numbers $s_{1,n}$ and $s_{2,n}$ satisfying  $s_{1, n} \leq s_{2, n}$, the inclusion
$$ A(t)\cap B(t) \ \subset \ A(s_{2, n})\cup B(s_{1, n})$$
holds simultaneously for all $t\in\R$.
Therefore, after taking the supremum over $t\in\R$, we have
\begin{equation}\label{eqn:afterunion}
\dK(\L(M_p(\hat Y)),\L(M_{k_n}(\hat Y))) \ \leq \ \P(A(s_{2, n}))+\P(B(s_{1, n})).
\end{equation}
Let 
\begin{align}
    & \omega=\frac{\epsilon}{24(\beta\vee 1)C},\label{eqn:omegadef}\\
    & d_n=\big\lfloor\ts\frac{\omega^2}{4} {\tt{r}}(R(l_n)) \vee 2\big\rfloor,\label{eqn:dndef}
\end{align}
where ${\tt{r}}(R(l_n)):=\textup{tr}(R(l_n))^2/\|R(l_n))\|_F^2=\frac{l_n^2}{\| R(l_n) \|_F^2}$ is the stable rank of $R(l_n)$. We will choose $s_{1,n}$ and $s_{2,n}$ according to 
\begin{equation}
\begin{aligned}
s_{1,n}&=c_1 k_n^{-\beta(1-\tau)/2} (\log(n)\vee 2),\\[0.2cm]
s_{2,n}&=c_2 l_n^{-\beta(1-\tau)} \sqrt{\log(d_n)},
\end{aligned}
\label{eqn:s1ns2ndef}
\end{equation}
for some constants $c_1,c_2>0$ not depending on $n$. It can also be checked that for any fixed choices of $c_1$ and $c_2$, the inequality $s_{1,n} \leq s_{2,n}$ holds for all large $n$  due to the definitions of $k_n, l_n$ and $d_n$.

To bound $\P(A(s_{2, n}))$, we have 
\begin{align*}
\P\big(A(s_{2, n}) \big) \ \leq \ & \P(\tilde M_{k_n}(X) \leq s_{2,n}) +d_{\mathrm{K}}(\mathcal{L}(M_{k_n}(\hat Y)), \mathcal{L}(\tilde{M}_{k_n}(X))),
\end{align*}
where the first term on the right hand side is of order $n^{-1/2}$ by Lemma \ref{Lemma: tilde MX}, and the second term is of order $n^{-\frac{1}{2}+\epsilon}$ by Propositions \ref{Prop: practical kn} and \ref{prop: Ykn-tildeXkn}. Lastly,  Lemma \ref{Lemma: B order} shows that $\P(B(s_{1, n}))$ is of order $\frac{1}{n}$, which completes the proof.\qed

\begin{lemma}
\label{Lemma: tilde MX}
Suppose conditions of Theorem~\ref{thm:main} hold.  Then, there is a constant $c_2>0$, not depending on $n$, such that the following bound holds when $s_{2,n}=c_2 l_n^{-\beta(1-\tau)} \sqrt{\log(d_n)}$,
$$
\P\Big(\tilde M_{k_n}(X) \leq s_{2,n}\Big) \lesssim n^{-1/2}.
$$
\end{lemma}

\noindent \emph{Proof.}
Observe that for any $t \in \R$,
\begin{align}\label{eqn:tildeknln}
\P(\tilde M_{k_n}(X) \leq t) \leq \P(\tilde M_{l_n}(X) \leq t). 
\end{align}
Let $(a_j)_{j \in J(l_n)}$ and $b$ be positive numbers with $\max_{j \in J(l_n)} a_j \leq b$. For any sequence of random variables $(U_j)_{j \in J(l_n)}$ and $t \geq 0$, it is straightforward to check that
\begin{align*}
\P\Big(\max_{j \in J(l_n)} U_j \leq t\Big) \leq \P\Big(\max_{j \in J(l_n)} a_j U_j \leq bt\Big).
\end{align*}
Consider the choice $a_j = \sigma_{j}^{-(1-\tau)}$ and note that under Assumption~\ref{A:model}(\ref{A:var}), there is a constant $c_0>0$ not depending on $n$ such that the bound $a_j\leq c_0 l_n^{(1-\tau)\beta}\!=:b$ holds for all $j\in J(l_n)$. So, if we let $U_j=G_j(X)/\sigma_{j}^\tau$, then the previous two displays imply
\begin{align*}
\P\big(\tilde M_{k_n}(X)\leq t\big) 
%
&\, \leq\, \P\Big(\max_{j\in J(l_n)} \frac{G_j(X)}{\sigma_{j}} \leq c_0  l_n^{(1-\tau)\beta} t \Big). 
\end{align*}
Consequently, if we let $\omega$ be as defined in~\eqref{eqn:omegadef}, and let $c_2=\frac{1}{c_0} \omega \sqrt{2(1-\omega)}$ in the definition~\eqref{eqn:s1ns2ndef} of $s_{2,n}$, then choosing $t=s_{2,n}$ in the previous display gives
$$\P\Big(\tilde M_{k_n}(X) \leq s_{2,n}\Big) \ \leq \ 
\P\Big(\max_{j\in J(l_n)} \frac{G_j(X)}{\sigma_{j}} \ \leq \omega \sqrt{2(1-\omega)\log(d_n)}\Big).$$
To bound the probability on the right, we apply Lemma~\ref{lem:lowertail} with $(l_n, d_n, \omega, \omega)$ playing the roles of $(d, k, a, b)$ in the statement of that result, which yields
\begin{equation*}
    \P\Big(\tilde M_{k_n}(X) \leq s_{2,n}\Big) \  \lesssim  \ d_n^{-(1-\omega)^3/\omega} (\log(d_n))^{\frac{1-\omega(2-\omega)-\omega}{2\omega}}.
\end{equation*}
(Note that the conditions of Lemma~\ref{lem:lowertail} are applicable because $2 \leq d_n \leq \frac{\omega^2}{4} {\tt{r}}(R(l_n))$ when $n$ is sufficiently large.)
Furthermore, by Assumption~\ref{A:model}(\ref{A:cor}), we have
\begin{align*}
d_n \asymp {\tt{r}}(R(l_n))
=\frac{l_n^2}{\|R(l_n) \|_F^2}  \ \gtrsim \  l_n^{\frac{1}{C}}  \ \gtrsim \ n^{\omega}. 
\end{align*}
Hence, there is a constant $c>0$ not depending on $n$ such that
\begin{equation*}
\begin{split}
    \P\Big(\tilde M_{k_n}(X) \leq s_{2,n}\Big) 
   & \ \lesssim \ d_n^{-\frac{(1-\omega )^3}{\omega }}\log(d_n)^c\\
    & \ \lesssim \ n^{-(1-\omega )^3}\log(n)^c\\
    & \ \lesssim \ n^{-1/2}
\end{split}
\end{equation*}
as needed.
\qed

\begin{lemma}
\label{Lemma: B order}
If the conditions of Theorem~\ref{thm:main} hold, then there is a constant $c_1>0$ not depending on $n$ such that the following bound holds when $s_{1,n}=c_1 k_n^{-\beta(1-\tau)/2} (2\vee \log(n))$,
\begin{align*}
\P\big(B(s_{1,n}) \big) \, \lesssim \, \ts \frac{1}{n}.
\end{align*}
\end{lemma}

\noindent \emph{Proof.}
Let $q=2\vee\log(n)$ and observe that
{\small
\begin{equation*}
\begin{split}
    \bigg\|\max_{j \in J(k_n)^c} \frac{\sum_{i=1}^n \hat  Y_{ij}}{\hat \sigma_j^{\tau} \sqrt{n}}\bigg\|_{L^q|X'}^q 
    \ \leq \ \sum_{j\in J(k_n)^c} \Bigg(2^q\bigg\|\frac{\sum_{i=1}^n \hat Y_{ij}-\E(\hat Y_{ij}|X')}{\hat\sigma_j^{\tau}\sqrt n}\bigg\|_{L^q|X'}^q 
    \ + \ 2^q\bigg|\frac{\sum_{i=1}^n \E(\hat Y_{ij}|X')}{\hat\sigma_j^{\tau}\sqrt n} \bigg|^q\Bigg)
    \end{split}
\end{equation*}
}
holds almost surely. By Rosenthal's inequality (Lemma \ref{Lemma: Rosenthal's inequality}), the following event holds almost surely,
\begin{equation*}
\begin{split}
    \bigg\|\frac{\sum_{i=1}^n \hat Y_{ij}-\E(\hat Y_{ij}|X')}{\hat\sigma_j^{\tau}\sqrt n}\bigg\|_{L^q|X'} & 
    \ \leq \ \frac{c q}{\hat\sigma_j^{\tau}}\max \bigg\{ \sqrt{\textstyle \var(\hat Y_{1j}|X')} \ , n^{-1/2+1/q} \|\hat Y_{1j}\|_{L^q|X'}\bigg\}\\[0.2cm]
    & \ \leq \ \frac{c q}{\hat\sigma_j^{\tau}}\max \Big\{ \sigma_j  \ , n^{1/q} \hat\sigma_j\Big\},
    \end{split}
\end{equation*}
where the second step follows from $\var(\hat Y_{1j}|X') \leq \E(X_{1j}^2) \leq \sigma_j^2$ and $|\hat Y_{1j}| \leq n^{1/2} \hat \sigma_j$. Lemmas \ref{Lemma: MOM variance}\eqref{Equa: MOM variance 2} and \ref{Lemma: MOM var anti-con} imply that 
there is an constant $c>0$ not depending on $n$ such that both of the bounds
\begin{equation*}
\max_{j\in J(k_n)^c}\frac{\hat\sigma_j^{1-\tau}}{\sigma_j^{(1-\tau)/2}}\leq c \ \ \ \text{ and }  \ \ \   \max_{j\in J(k_n)^c}\frac{\sigma_j}{\hat\sigma_j^{\tau} \sigma_j^{(1-\tau)/2}} \leq c
\end{equation*}
hold simultaneously with probability at least $1-cn^{-(2+\delta)}$ . Consequently, the bound
\begin{equation*}
\begin{split}
      \bigg\|\frac{\sum_{i=1}^n \hat Y_{ij}-\E(\hat Y_{ij}|X')}{\hat\sigma_j^{\tau}\sqrt n}\bigg\|_{L^q|X'}
      & \leq c q\max\bigg\{ \sigma_j^{(1-\tau)/2} ,\ \sigma_j^{(1-\tau)/2} n^{1/q} \bigg\}\\
      &\leq cq \sigma_j^{(1-\tau)/2}
      \end{split}
\end{equation*}
holds simultaneously over all $j\in J(k_n)^c$ with probability at least $1-cn^{-(2+\delta)}$. 
Using Lemma \ref{Lemma: moments closeness} and similar reasoning, it can also be shown that
\begin{equation*}
\Big|\frac{\E(\hat Y_{ij}|X')}{\hat \sigma_j^{\tau}\sqrt{n}}\Big| \leq c \sigma_j^{(1-\tau)/2}n^{-2} 
\end{equation*}
holds simultaneously over all $j\in J(k_n)^c$ with probability at least $1-cn^{-(2+\delta)}$. Combining the last several steps and Assumption \ref{A:model}(\ref{A:var}), we conclude that the bound
\begin{equation*}
\begin{split}
     \bigg\|\max_{j \in J(k_n)^c} \frac{\sum_{i=1}^n \hat  Y_{ij}}{\hat \sigma_j^{\tau} \sqrt{n}}\bigg\|_{L^q|X'}^q 
     & \ \leq \ (cq)^q \sum_{j\geq k_n} j^{-q \beta (1-\tau)/2 }\\[0.2cm]
     &\ \leq \ c^q \frac{q^q}{(1-\tau)q\beta/2-1} k_n^{1-(1-\tau)q\beta /2}
     \end{split}
\end{equation*}
holds with probability at least $1-cn^{-(2+\delta)}$, where in the second step we have used the fact that $q\beta(1-\tau) > 2$ when $n$ is sufficiently large. Also, since $q \asymp \log(n)$, we have
\begin{align*}
\Big(\frac{1}{(1-\tau)q\beta/2-1}\Big)^{1/q} \lesssim 1.
\end{align*}
Thus, there is a sufficiently large choice of $c_1>0$, such that if $s_{1,n}=c_1 q k_n^{-\beta(1-\tau)/2}$, then the bound
\begin{align*}
\P(B(s_{1,n})|X') \, \leq \, \frac{1}{s_{1,n}^q}     \bigg\|\max_{j \in J(k_n)^c} \frac{\sum_{i=1}^n \hat  Y_{ij}}{\hat \sigma_j^{\tau} \sqrt{n}}\bigg\|_{L^q|X'}^q \, \leq \,\ts e^{-q} \, \lesssim \, \ts\frac{1}{n}
\end{align*}
holds with probability at least $1-cn^{-(2+\delta)}$. This implies the stated result.\qed

\begin{proposition}
\label{Prop: practical kn}
If the conditions of Theorem~\ref{thm:main} hold, then
\begin{equation*}
\dK\Big(\L(M_{k_n}(\hat Y)),\L(M_{k_n}(Y))\Big) \lesssim n^{-\frac{1}{2}+\epsilon}.
\end{equation*}
\end{proposition}

\noindent \emph{Proof.} Using the decomposition~\eqref{eqn:anticouple} for the Kolmogorov metric, followed by the bound for anti-concentration probabilities in~\eqref{eqn:anticonc_decomp}, we have
\begin{align*}
\dK(\L(M_{k_n}(\hat Y)),\L(M_{k_n}(Y))) \leq & \  \sup_{t \in \R} \P \big(|\tilde M_{k_n}(X)-t| \leq n^{-\frac{1}{2}+\frac{3\epsilon}{4}} \log(n)\big)\\
& \ + \  2\dK(\L(\tilde M_{k_n}(X)),\L(M_{k_n}(Y))) \\
& \ + \  \P\Big(\big|M_{k_n}(\hat Y) - M_{k_n}(Y)\big| \geq n^{-\frac{1}{2}+\frac{3\epsilon}{4}} \log(n) \Big).
\end{align*}
For the first term on the right side, Nazarov's inequality (Lemma \ref{Lemma: Nazarov's inequality}) and Assumption 1(\ref{A:var}) imply 
\begin{equation}\label{Equa: MknX nazarov}
\begin{split}
\sup_{t\in\R} \P \Big(| \tilde M_{k_n}(X)-t| \leq n^{-\frac{1}{2}+\frac{3\epsilon}{4}} \log(n)\Big)
&  \lesssim \, n^{-\frac{1}{2}+\frac{3\epsilon}{4}} \log(n) k_n^{\beta(1-\tau)} \sqrt{\log(k_n)} \\
& \lesssim \ n^{-\frac{1}{2}+\epsilon}.
\end{split}
\end{equation}
Next, it follows from Proposition \ref{prop: Ykn-tildeXkn} that
$$ \dK\Big(\L(\tilde M_{k_n}(X)),\L(M_{k_n}(Y))\Big) \ \lesssim \ n^{-\frac{1}{2}+\epsilon}.$$
Finally, Lemma~\ref{Lemma: hat Y and Y Mkn} implies
$$ \P\Big(\big|M_{k_n}(\hat Y) - M_{k_n}(Y)\big| \geq n^{-\frac{1}{2} + \frac{\epsilon}{2}} \log(n) \Big) \lesssim \ n^{-\frac{1}{2}+\epsilon},$$
which completes the proof.
\qed

\begin{lemma}\label{Lemma: hat Y and Y Mkn}
If the conditions of Theorem~\ref{thm:main} hold, then
\begin{align*}
\P\Big(\big|M_{k_n}(\hat Y) - M_{k_n}(Y)\big| \geq  n^{-1/2+\epsilon/2} \log(n)\Big) 
\lesssim n^{-1/2+\epsilon}.     
\end{align*}
\end{lemma}

\noindent \emph{Proof.} First observe that $\big|M_{k_n}(\hat Y) - M_{k_n}(Y)\big|$ can be bounded by
\begin{equation}\label{eqn:hatY-Y-coupling}
\begin{split}
    & \max_{j \in J(k_n)}\!\bigg|\frac{\sum_{i=1}^n \hat Y_{ij}-Y_{ij}}{\hat \sigma_j^{\tau} \sqrt{n}}\bigg|  
    + \max_{j \in J(k_n)}\!\bigg|\frac{\sum_{i=1}^n Y_{ij}}{\hat \sigma_{j}^{\tau} \sqrt n} -\frac{\sum_{i=1}^n Y_{ij}}{\sigma_{j}^{\tau} \sqrt n} \bigg|  
    +  \max_{j\in J(k_n)}\!\frac{\sqrt n |\E(Y_{1j})|}{\sigma_j^{\tau}}.
    \end{split}
\end{equation}
The first term in the bound \eqref{eqn:hatY-Y-coupling} is $0$ with probability at least $1-\frac{ck_n}{n}$ by Lemma \ref{Lemma: Y hatY diff}, and the (deterministic) third term is $\mathcal{O}(n^{-1})$ by Lemma \ref{Lemma: moments closeness}.

It remains to handle the middle term in the bound \eqref{eqn:hatY-Y-coupling}, which satisfies 
\begin{align}
\label{eqn:twofactors}
\max_{j \in J(k_n)} \big|\frac{\sum_{i=1}^n Y_{ij}}{\hat \sigma_{j}^{\tau} \sqrt n} -\frac{\sum_{i=1}^n Y_{ij}}{\sigma_{j}^{\tau} \sqrt n} \big| 
\leq \max _{j \in J(k_n)}\Big|\Big(\frac{{\sigma}_{j}}{\hat \sigma_{j}}\Big)^{\tau}-1\Big| \cdot \max_{j\in J(k_n)} \bigg|\frac{\sum_{i=1}^n Y_{ij}}{\sqrt{n} \sigma_j^{\tau}}\bigg|.
\end{align}
Since $|a^{\tau}-1|\leq |a^2-1|$ for any $a\geq 0$ and $\tau\in [0,1)$, the first factor on the right is of order $n^{-1/2+\epsilon/2}$ with probability at least $1-cn^{-(2+\delta)}$ by Lemma \ref{Lemma: MOM variance}\eqref{Equa: MOM mean 1 sum}.

Next we will show the second factor in the bound~\eqref{eqn:twofactors} is of order $\log(n)$ with probability at least $1-\frac{c}{n}$. Let $q=2\vee\log(n )$ and observe that under Assumption \ref{A:model}(\ref{A:var}), Lemma \ref{Lemma: moments closeness} implies
$$
\bigg\| \max_{j \in J(k_n)} \Big|\frac{\sum_{i=1}^n Y_{ij} }{\sqrt{n} \sigma_j^{\tau}} \Big| \bigg\|_{L^q} \ \lesssim \  \bigg\| \max_{j \in J(k_n)} \Big|\sum_{i=1}^n \frac{ Y_{ij} - \E(Y_{ij})}{\sqrt{n} \sigma_{j}^{\tau}} \Big| \bigg\|_{L^q} + \frac{1}{n}.
$$
Furthermore, Lemma \ref{Lemma: Lq of Snj} gives
\begin{align*}
\bigg\| \max_{j \in J(k_n)} \bigg|\sum_{i=1}^n \frac{ Y_{ij} - \E(Y_{ij})}{\sqrt{n}\sigma_{j}^{\tau}} \bigg| \bigg\|_{L^q}^q 
& \leq \sum_{j \in J(k_n)} \Big\| \sum_{i=1}^n \frac{ Y_{ij} - \E(Y_{ij})}{\sqrt{n} \sigma_{j}^{\tau}} \Big\|_{L^q}^q \\
&  \leq c^q q^q  \sum_{j \in J(k_n)} \sigma_{j}^{(1-\tau)q}\\
& \leq c^q q^q.
\end{align*}
Therefore, Chebyshev's inequality implies that the second factor in the bound~\eqref{eqn:twofactors} is of order $\log(n)$ with probability $1-\frac{c}{n}$.
\qed

\begin{lemma}\label{Lemma: Y hatY diff}
If the conditions of Theorem~\ref{thm:main} hold, then
$$
\P \Big( \max_{j \in J(k_n)} \max_{1\leq i\leq n} |Y_{ij}-\hat Y_{ij}| > 0 \Big) \, \lesssim \, \frac{k_n}{n}.
$$
\end{lemma}

\noindent \emph{Proof.} First observe that for each $i$ and $j$, the definitions of $Y_{ij}$ and $\hat Y_{ij}$ give
$$ \P \big(|Y_{ij}-\hat Y_{ij}| > 0 \big) \ \leq \ \P\Big(|X_{ij}|>\min\{t_j,\hat t_j\}\Big).$$
Recalling that $\hat t_j=\hat\sigma_j\sqrt n$ and $t_j=\sigma_j\sqrt n$, the events $\{\hat t_j > t_j/2\}$ for $j\in J(k_n)$ occur simultaneously with probability at least $1-cn^{-(2+\delta)}$ by Lemma \ref{Lemma: MOM variance}\eqref{Equa: MOM variance 1 sum}. Combining this with a union bound over $i=1,\dots,n$, we have 
\begin{equation}
\begin{aligned}
     \P \Big(\max_{1\leq i\leq n}|Y_{ij}-\hat Y_{ij}| > 0 \Big) 
     & \ \lesssim \ \sum_{i=1}^n\Big( \P\big(|X_{1j}|> \textstyle\frac{t_j}{2}\big) \ + \ n^{-(2+\delta)} \Big)\\
     & \ \lesssim \ \frac{n\,\E|X_{1j}|^4}{t_j^4}+ \ n^{-(1+\delta)}\\[0.2cm]
     &\ \lesssim \, \frac{1}{n},
\end{aligned} 
\label{equa: YhatY positive}
\end{equation}
where the last step uses Assumption~\ref{A:model}(\ref{A:moments}). Finally, taking a union bound over $j\in J(k_n)$,
$$
\P \Big( \max_{j \in J(k_n)} \max_{1\leq i\leq n} |Y_{ij}-\hat Y_{ij}| > 0 \Big) \lesssim \frac{k_n}{n}
$$
as needed.
\qed

\begin{lemma}\label{lem:couple_MkY_MkX}
If the conditions of Theorem~\ref{thm:main} hold, then
\begin{equation*}
\P\Big(|M_{k_n}(Y)-M_{k_n}(X)| \geq  n^{-\frac{1}{2}}\Big)
\ \lesssim \  n^{-\frac{1}{2}+\epsilon},
\end{equation*}
\end{lemma}

\noindent\emph{Proof.} For each $i=1,\dots,n$ and $j \in J(k_n)$, let
$$\Delta_{ij}= \frac{1}{\sigma_{j}^{\tau}\sqrt n}\Big(Y_{ij}-\E(Y_{ij})-X_{ij}\Big),$$ 
so that
\begin{align*}
|M_{k_n}(Y)-M_{k_n}(X)| \ \leq \ \max_{j\in J(k_n)}\Big|\sum_{i=1}^n \Delta_{ij}\Big|.  
\end{align*}
Noting that $Y_{ij}$ and $X_{ij}$ only differ when $|X_{ij}|> t_j$, we have
\begin{align*}
\E \big(|\Delta_{ij}|\big) 
& \ \lesssim \  \E\Big({\ts\frac{1}{\sqrt{n} \sigma_j^{\tau}}} |X_{ij}| 1\{|X_{ij}|>t_j\}\Big) + \frac{1}{\sqrt{n}\sigma_j^{\tau}} \big|\E(Y_{ij})\big| \\
 & \ \lesssim \ \frac{1}{\sqrt{n} \sigma_j^{\tau}}\|X_{ij}\|_{L^{\frac{4}{1+\epsilon}}} \|1\{|X_{ij}|>t_j\}\|_{L^{\frac{4}{3-\epsilon}}} + n^{-2}\\
& \ \lesssim \ n^{-2+\frac{\epsilon}{2}}.
\end{align*}
where we have used H\"older's inequality and Lemma \ref{Lemma: moments closeness} in the first step, followed by Chebyshev's inequality in bounding  $\|1\{|X_{ij}|>t_j\}\|_{L^{\frac{4}{3-\epsilon}}}$. Therefore,
\begin{align*}
\P\Big(\max_{j\in J(k_n)}\big|\sum_{i=1}^n \Delta_{ij}\big| \geq n^{-\frac{1}{2}}\Big) \
& \ \leq \ \sum_{j\in J(k_n)}  \P\Big(\big|\sum_{i=1}^n \Delta_{ij}\big| \geq n^{-\frac{1}{2}}\Big) \\
& \ \lesssim \ k_n n^{-\frac{1}{2}+\frac{\epsilon}{2}}\\
& \ \lesssim \ n^{-\frac{1}{2}+\epsilon},
\end{align*}
as needed. \qed

\begin{proposition}
\label{prop: Ykn-tildeXkn}
If the conditions of Theorem~\ref{thm:main} hold, then
$$ \dK(\L(M_{k_n}(Y)), \L(\tilde M_{k_n}(X))) \ \lesssim \ n^{-\frac{1}{2}+\epsilon}.$$
\end{proposition}
\noindent\emph{Proof.} First observe that
\begin{align*}
\dK(\L(M_{k_n}(Y)), \L(\tilde M_{k_n}(X)))  \ \leq \ \dK(\L(M_{k_n}(Y)), \L(M_{k_n}(X))) + \dK(\L(M_{k_n}(X)), \L(\tilde M_{k_n}( X))).
\end{align*}
The second term on the right side is of $n^{-1/2+\epsilon}$ by Lemma \ref{Lemma: Gaussian app X}. Using the decomposition~\eqref{eqn:anticouple} for the Kolmogorov metric, the first term on the right side can be bounded by
\begin{align*}
\dK(\L(M_{k_n}(Y)), \L(M_{k_n}(X)))  \leq  \sup _{t \in \mathbb{R}} \P(|M_{k_n}(X)-t| \leq n^{-\frac{1}{2}})  +  \P(|M_{k_n}(X)- M_{k_n}(Y)|>n^{-\frac{1}{2}}).
\end{align*}
Lemma \ref{lem:couple_MkY_MkX} shows that 
$$
\P(|M_{k_n}(X)- M_{k_n}(Y)|>n^{-\frac{1}{2}}) \lesssim n^{-\frac{1}{2}+\epsilon}.
$$
Using the bound for anti-concentration probabilities in~\eqref{eqn:anticonc_decomp}, we have
\begin{align*}
\sup _{t \in \mathbb{R}} \P(|M_{k_n}(X)-t| \leq n^{-\frac{1}{2}}) \leq & \ \sup _{t \in \mathbb{R}} \P(|\tilde M_{k_n}(X)-t| \leq n^{-\frac{1}{2}}) \\
& + 2\dK(\L(M_{k_n}(X)), \L(\tilde M_{k_n}( X))).
\end{align*}
Nazarov's inequality (Lemma \ref{Lemma: Nazarov's inequality}) and Assumption 1(\ref{A:var}) imply 
\begin{equation}
\begin{split}
\sup_{t\in\R} \P (| \tilde M_{k_n}(X)-t| \leq n^{-\frac{1}{2}})) \lesssim \ n^{-\frac{1}{2}+\epsilon},
\end{split}
\end{equation}
which completes the proof.

\begin{lemma}\label{Lemma: Gaussian app X}
If the conditions of Theorem~\ref{thm:main} hold, then
$$\dK(\L(M_{k_n}(X)),\L(\tilde M_{k_n}(X))) \lesssim n^{-1/2+\epsilon}.$$
\end{lemma}
\noindent \emph{Proof.} For each $i=1,\dots,n$, let $X_i(k_n)$ denote the vector in $\R^{k_n}$ corresponding to the coordinates of $X_i$ indexed by $J(k_n)$, and let $\mathcal{R}_t=\prod_{j\in J(k_n)}(-\infty, t\sigma_j^{\tau}]$, so that 
$$\P(M_{k_n}(X)\leq t) = \P\Big(\ts \frac{1}{\sqrt n}\sum_{i=1}^n X_i(k_n) \in \mathcal{R}_t\Big).$$
If the rank of the covariance matrix of $X_1(k_n)$ is denoted by $r$, let $\Pi_r\in\R^{k_n\times r}$ be the  matrix whose columns correspond to the leading $r$ eigenvectors of the covariance matrix of $X_1(k_n)$. In particular, we have $X_1(k_n)=\Pi_r\Pi_r\ttop X_1(k_n)$ almost surely, and it follows that
 $$\P(M_{k_n}(X)\leq t) = \P\Big(\ts \frac{1}{\sqrt n}\sum_{i=1}^n \Pi_r\ttop X_i(k_n)  \in \Pi_r^{-1}(\mathcal{R}_t)\Big),$$
 where $\Pi_r^{-1}(\mathcal{R}_t)$ refers to the pre-image. Next, the definition of $\Pi_r$ ensures that the covariance matrix of $\Pi_r\ttop X_1(k_n)$, denoted by $\mathfrak{S}_r$, is invertible. So, if we define the random vector
 $$ V_i = \mathfrak{S}_r^{-1/2}\Pi_r\ttop  X_i(k_n)$$
 for each $i=1,\dots,n$ and the set
 $\mathcal{C}_t= \mathfrak{S}_r^{-1/2}\Pi_r^{-1}(\mathcal{R}_t)$ for each $t\in\R$, then
$$\P(M_{k_n}(X)\leq t)= \P\Big(\ts\frac{1}{\sqrt n}\sum_{i=1}^n V_i \in \mathcal{C}_t\Big).$$
It can also be shown by similar reasoning that
$$
\P(\tilde M_{k_n}(X)\leq t)= \gamma_{r}(\mathcal{C}_t),
$$
where $\gamma_r$ denotes the standard Gaussian distribution on $\R^r$.
Due to the fact that the i.i.d.~random vectors $V_1,\dots,V_n$ are centered and isotropic, Bentkus' multivariate Berry-Esseen theorem (Lemma~\ref{Lemma: Bentkus berry esseen}) ensures there is an absolute constant $c>0$ such that
\begin{equation}\label{eqn:applyBentkus}
\begin{split}
\dK(\L(M_{k_n}(X)),\L(\tilde M_{k_n}(X)))
&\leq \ 
\sup_{\mathcal{C}\in\mathscr{C}}\Big|\P\Big(\ts\frac{1}{\sqrt n}\sum_{i=1}^n V_i \in \mathcal{C}\Big)-\gamma_r(\mathcal{C})\Big| \\[0.2cm]
 & \leq \ \displaystyle \frac{c r^{1/4} \E\|V_1\|_2^3}{\sqrt n},
\end{split}
\end{equation}
where $\mathscr{C}$ denotes the collection of all Borel convex subsets of $\R^{r}$. To bound the third moment on the right side, observe that
\begin{align*}
\E\|V_1\|_2^3 & \ = \  \Big\| \ts\sum_{j=1}^r V_{1j}^2 \Big\|_{L^{3/2}}^{3/2}\ \leq \ \Big( \ts\sum_{j=1}^r \|V_{1j}\|_{L^{3}}^{2}\Big)^{3/2}.
\end{align*}
Since $V_{1j}$ can be expressed as $\langle v, X_1\rangle$ for some vector $v\in\R^p$, Assumption~\ref{A:model}(\ref{A:moments}) implies $\|V_{1j}\|_{L^3}^2\lesssim \var(V_{1j})=1$ for each $j=1,\dots,r$. Thus $\E\|V_1\|_2^3\lesssim r^{3/2}\leq k_n^{3/2}$, and combining with~\eqref{eqn:applyBentkus} completes the proof.\qed

\section{Gaussian comparison}\label{sec:comp}
Recall that $M_{k_n}^*(X) = \max_{j\in J(k_n)}\frac{1}{\sqrt{n}\sigma_j^{\tau}}\sum_{i=1}^n \xi_i (X_{ij}-\bar X_j)$ from \eqref{equa: MdstarX}.
\begin{proposition}
\label{prop: gaus-compa}
Suppose the conditions of Theorem~\ref{thm:main} hold. Then, there is a constant $c > 0$ not depending on $n$ such that the event
\begin{equation*}
\dK\Big(\L(\tilde M_{k_n}(X))\,,\,\L( M_{k_n}^*(X)|X)\Big) \leq c n^{-\frac{1}{2}+\epsilon} 
\end{equation*}
holds with probability at least $1-cn^{-\delta/4}$.     
\end{proposition}

\noindent\emph{Proof.} Let $V_1,\dots,V_n\in\R^r$ be as in the proof of Lemma \ref{Lemma: Gaussian app X} and put $\bar V =\frac{1}{n} \sum_{i=1}^n V_i$. The reasoning used in the proof of Lemma~\ref{Lemma: Gaussian app X} shows that 
for any $t\in\R$, there is a convex Borel set $\mathcal{C}_t\subset\R^r$ such that 
\begin{equation*}
    \P\big(M_{k_n}^*(X)\leq t | X\big) \ = \ \P\Big(\frac{1}{\sqrt n}\sum_{i=1}^n \xi_i(V_i - \bar V)\in\mathcal{C}_t \Big| X\Big),
\end{equation*}
and also, that
\begin{equation*}
        \P(\tilde M_{k_n}(X)\leq t) = \gamma_r(\mathcal{C}_t),
\end{equation*}
where $\gamma_r$ is the standard Gaussian distribution on $\R^r$. Next, define the sample covariance matrix 
\begin{equation*}
    W_{r} = \frac{1}{n}\sum_{i=1}^n (V_i-\bar V)(V_i-\bar V)\ttop
\end{equation*}
and observe that Lemma~\ref{lem:GaussianFrobenius} gives the following almost-sure bound,
\begin{equation*}
    \sup_{t\in\R}\Big|\P\big(M_{k_n}^*(X)\leq t\big|X\big) \ - \ \P\big(\tilde M_{k_n}(X)\leq t\big)\Big| \ \leq \ 2 \big\| W_{r}-I_{r}\big\|_F,
\end{equation*}
where $I_r$ denotes the identity matrix of size $r \times r$. To handle the Frobenius norm, Assumption~\ref{A:model}(\ref{A:moments}) implies
\begin{align*}
 & \max_{1\leq j,j'\leq r}\E\Big| e_j\ttop (V_1V_1\ttop  - I_{r})e_{j'}\Big|^{\frac{4+\delta}{2}} \lesssim 1, \\
 &  \max_{1\leq j\leq r} \E\big| e_j\ttop V_1 \big|^{4+\delta} \lesssim 1.
\end{align*}
Consequently, the Fuk-Nagaev inequality (Lemma \ref{Lemma: Fuk-Nagaev inequalities}) with the choices $q=(4+\delta)/2$ and $q=4+\delta$ in the notation of that result ensures that for each $j, j' = 1,\ldots,r$,
\begin{align*}
& \P\bigg(\Big|\frac{1}{n}\sum_{i=1}^n e_j\ttop (V_iV_i\ttop  - I_{r})e_{j'}\Big| \geq 
   \ts n^{-1/2+\epsilon/2} \bigg) \ \lesssim  \ n^{-(\delta/4+\epsilon)},\\
& \P\Big( \big|e_j\ttop \bar V \big| \geq n^{-1/4+\epsilon/4} \Big) \leq n^{-(2+3\delta/4+\epsilon)}.
\end{align*}
So, using the identity
$$W_r-I_r=\Big(\frac{1}{n}\sum_{i=1}^n V_iV_i\ttop -I_r\Big)-\bar V\bar V\ttop$$
it is straightforward to check that the event
\begin{equation*}
    \big\| W_{r}-I_{r}\big\|_F \geq 2k_n n^{-1/2+\epsilon/2}
\end{equation*}
holds with probability at most of order $k_n^2 n^{-(\delta/4+\epsilon)} \lesssim n^{-\delta/4}$, which leads to the stated result.\qed

\section{Bootstrap approximation}\label{sec:boot}
\begin{proposition}
\label{prop: XZ kn-p}
Suppose the conditions of Theorem~\ref{thm:main} hold. Then, there is a constant $c>0$ not depending on $n$ such that the event
$$\dK\Big(\L(M_{k_n}^*(X) | X),\L( M_{p}^*(\hat Z)|X))\Big) \leq cn^{-\frac{1}{2}+\epsilon} $$
holds with probability at least $1-c n^{-\delta/4}$.    
\end{proposition}

\noindent\emph{Proof.} 
Consider the inequality
\begin{align*}
\dK\Big(\L(M_{k_n}^*(X) | X),\L( M_{p}^*(\hat Z)|X))\Big) \leq \mathrm{I}'_n + \mathrm{II}'_n, 
\end{align*}
where we define
\begin{align*}
 \mathrm{I}'_n &= \dK\Big(\L(M_{k_n}^*(X) |X ),\L( M_{k_n}^*(\hat Z)|X)\Big), \\[0.2cm]
 \mathrm{II}'_n &= \dK\Big(\L( M_{k_n}^*(\hat Z)|X), \L(M_{p}^*(\hat Z)|X)\Big).
\end{align*}
Both $\mathrm{I}'_n$ and $\mathrm{II}'_n$ are of order at most $n^{-1/2+\epsilon}$ with probability at least $1-cn^{-\delta/4}$, as shown in Lemma \ref{lemma: XZ kn-kn} and Proposition \ref{prop: hatZ kn-p}.\qed

\begin{lemma}
\label{lemma: XZ kn-kn}
Suppose the conditions of Theorem~\ref{thm:main} hold. Then, there is a constant $c>0$ not depending on $n$ such that the event
$$\dK\Big(\L(M_{k_n}^*(X) | X),\L( M_{k_n}^*(\hat Z)|X)\Big) \leq cn^{-1/2+\epsilon} $$
holds with probability at least $1-cn^{-\delta/4}$.
\end{lemma}

\noindent\emph{Proof.} 
The coupling and anti-concentration decomposition in~\eqref{eqn:anticouple} shows that for any $\eta>0$, we have
\begin{equation}\label{eqn:M*XM*Z}
\begin{aligned}
\dK\Big(\L\big(M_{k_n}^*(X) | X\big),\L\big( M_{k_n}^*(\hat Z)|X\big)\Big) \ \leq \ & 
\sup_{t\in\R}\P\Big(|M_{k_n}^*(X)-t| \leq \eta \Big|X\Big) \\
& +  \P\Big(\big|M_{k_n}^*(X)-M_{k_n}^*(\hat Z) \big|\geq \eta \Big|X\Big).
\end{aligned} 
\end{equation}
We will take $\eta= c\log(n) n^{-\frac{1}{2}+\frac{3\epsilon}{4}}$ for some constant $c>0$ not depending on $n$. Using the generic bound for anti-concentration probabilities in~\eqref{eqn:anticonc_decomp}, the first term on the right side of \eqref{eqn:M*XM*Z} is upper bounded by
$$
\sup_{t\in\R}\P\big(|\tilde M_{k_n}(X)-t| \leq \eta \big) + \ 2\dK\Big(\L \big( \tilde M_{k_n}(X)\big),\L\big( M_{k_n}^*(X)|X\big)\Big),
$$
which is at most of order $n^{-1/2+\epsilon}$ with probability at least $1-cn^{-\delta/4}$, due to~\eqref{Equa: MknX nazarov} and Proposition~\ref{prop: gaus-compa}.

To address the coupling term in~\eqref{eqn:M*XM*Z}, observe that for any $q\geq 2$, the basic inequalities~\eqref{eqn:gaussianLq}, \eqref{eqn:maxout} and \eqref{eqn:maxdiff} ensure there is a constant $c>0$ not depending on $n$ such that
\begin{equation*}
    \big\|M_{k_n}^*(X)-M_{k_n}^*(\hat Z)\big\|_{L^q|X} 
\ \leq \ c\sqrt q k_n^{1/q} \max_{j\in J(k_n)}\bigg(\frac{1}{n}\sum_{i=1}^n \Big(\textstyle\frac{X_{ij}-\bar X_j}{\sigma_j^{\tau}}-\frac{\hat Z_{ij}-\hat Z_j}{\hat\sigma_j^{\tau}}\Big)^2\bigg)^{1/2}.
\end{equation*}
To decompose this bound, define the random variables
\begin{align*}
    T_1& =  \displaystyle\max_{j\in J(k_n)}\bigg(\frac{1}{n}\sum_{i=1}^n \Big(\textstyle\frac{X_{ij}-\hat Z_{ij}}{\sigma_j^{\tau}}\Big)^2\bigg)^{1/2}, \\[0.2cm]
    T_2 &=  \displaystyle\max_{j\in J(k_n)} \Big|\frac{1}{\sigma_j^{\tau}}-\frac{1}{\hat\sigma_j^{\tau}} \Big| \bigg(\frac{1}{n}\sum_{i=1}^n \hat Z_{ij}^2\bigg)^{1/2}.
\end{align*}
Using Jensen's inequality for sample averages, it follows that
\begin{equation*}
    \big\|M_{k_n}^*(X)-M_{k_n}^*(\hat Z)\big\|_{L^q|X} \ \leq \  c \sqrt{q}k_n^{1/q} (T_1+T_2).
\end{equation*}
With regard to $T_1$, note that
\begin{equation*}
\begin{split}
    |X_{ij}-\hat Z_{ij}| & \ \leq \  |\tilde X_{j}|1\{|X_{ij}-\tilde X_j|\leq \hat t_j\}
     \ + \  (|X_{ij}|+\hat t_j) 1\big\{|X_{ij}-\tilde X_j|> \hat t_j\big\}.
    \end{split}
\end{equation*}
Also, by Lemma~\ref{Lemma: MOM variance}\eqref{Equa: MOM variance 1 sum}, the events
\begin{align*}
    \max_{j\in J(k_n)}\frac{\hat t_j}{t_j}  \ \leq \ 2 \ \ \ \ \ \text{ and } \ \ \ \  \ \min_{j\in J(k_n)} \frac{ \hat t_j}{t_j} \ \geq \ \textstyle\frac{1}{2}
\end{align*}
hold simultaneously with probability at least $1-cn^{-(2+\delta)}$ for some constant $c>0$ not depending on $n$. Based on this and $\min_{j\in J(k_n)}\sigma_j^{\tau}\gtrsim k_n^{-\beta\tau}$ under Assumption \ref{A:model}(\ref{A:var}), the bound
\begin{equation*}
    \Big(\textstyle\frac{X_{ij}-\hat Z_{ij}}{\sigma_j^{\tau}}\Big)^2  \ \leq \  ck_n^{2\beta\tau}\tilde X_{j}^2
     \ + \  ck_n^{2\beta\tau}(X_{ij}^2+t_j^2)\Big( 1\big\{|X_{ij}|> \textstyle\frac{t_j}{4}\big\} + 1\big\{ |\tilde X_j|>\textstyle\frac{t_j}{4}\big\}\Big)
\end{equation*}
holds simultaneously for all $i\in\{1,\dots,n\}$ and $j\in J(k_n)$ with probability at least $1-cn^{-(2+\delta)}$. Therefore, the bound
{\small
\begin{equation*}
    T_1 \leq  ck_n^{\beta \tau} \max_{j\in J(k_n)}\Bigg( |\tilde X_j| + \bigg( \frac{1}{n}\sum_{i=1}^n \Big(X_{ij}^2 +t_j^2 \Big) 1\{|X_{ij}|>\textstyle\frac{t_j}{4}\}\bigg)^{\frac{1}{2}} + 1\{|\tilde X_j|>\textstyle\frac{t_j}{4}\}\bigg(\displaystyle\frac{1}{n}\sum_{i=1}^n X_{ij}^2+t_j^2\bigg)^{\frac{1}{2}}\Bigg)
\end{equation*}}
holds with the same probability. 
The terms on the right side are handled as follows. First, $\max_{j \in J(k_n)}|\tilde X_j|$ is of order $n^{-1/2+\epsilon/2}$ with probability at least $1-cn^{-(2+\delta)}$ by Assumption \ref{A:model}(\ref{A:var}) and Lemma \ref{Lemma: MOM mean}\eqref{Equa: MOM mean 1 sum}. The indicators $1\{|X_{ij}|\geq \frac{t_j}{4}\}$ and $1\{|\tilde X_{j}|\geq \frac{t_j}{4}\}$  are 0 with probability at least $1-cn^{-2}$ due to the argument associated with the bounds in \eqref{equa: YhatY positive}, as well as Lemma \ref{Lemma: MOM mean}\eqref{Equa: MOM mean 1 sum} and the fact that $\sigma_{j} n^{-\frac{1}{2}+\frac{\epsilon}{2}} \lesssim \frac{t_j}{4}$. After taking a union bound over $i\in\{1,\dots,n\}$ and $j \in J(k_n)$, the event 
\begin{align*}
T_1 &\leq c k_n^{\beta\tau} n^{-1/2+\epsilon/2} \leq c n^{-1/2+3\epsilon/4}
\end{align*}
holds with probability at least $1-c k_n/n$.

Turning our attention to $T_2$, Lemma \ref{Lemma: MOM variance}\eqref{Equa: MOM variance 1 sum} implies that 
the bound
\begin{equation*}
    \displaystyle\max_{j\in J(k_n)} \Big|\frac{1}{\sigma_j^{\tau}}-\frac{1}{\hat\sigma_j^{\tau}} \Big| \ \leq ck_n^{\beta \tau}n^{-1/2+\epsilon/2} \leq cn^{-1/2+3\epsilon/4},
\end{equation*}
holds with probability at least $1-cn^{-(2+\delta)}$. Also, it will be shown in Equation \eqref{equa: sum hatZij} that the bound
$$ \max_{j\in J(k_n)} \Big(\frac{1}{n}\sum_{i=1}^n \hat Z_{ij}^2\Big)^{1/2} 
\leq  c \log(n) $$
holds with probability at least $1-ck_n/n$. Combining the last two steps shows that $T_2$ is of order $\log(n)n^{-1/2+3\epsilon/4}$ with the same probability.  \qed

\begin{proposition}
\label{prop: hatZ kn-p}
Suppose the conditions of Theorem~\ref{thm:main} hold. Then, there is a constant $c>0$ not depending on $n$ such that the event
$$
\dK(\L( M_{k_n}^*(\hat Z)|X),\L(M_{p}^*(\hat Z)|X)) \leq cn^{-\frac{1}{2} +\epsilon},
$$ 
holds with probability at least $1-cn^{-\delta/4}$.
\end{proposition}

\noindent \emph{Proof.} We may assume without loss of generality that $k_n<p$, for otherwise the quantity $\dK(\L(M_p(\hat Y)),\L(M_{k_n}(\hat Y)))$ is zero. For any $t\in\R$, define the events
\begin{align*}
A'(t)=\Big\{\max_{j\in J(k_n)}\frac{\sum_{i=1}^n \xi_i (\hat Z_{ij}-\hat  Z_j)}{\sqrt{n}\hat\sigma_j^{\tau}} \leq t\Big\} \ \  \  \text{and}   \ \ \  B'(t)=\Big\{\max_{j\in J(k_n)^c}\frac{\sum_{i=1}^n \xi_i (\hat Z_{ij}-\hat  Z_j)}{\sqrt{n}\hat\sigma_j^{\tau}} > t\Big\}.
\end{align*}
Using the argument in the proof of Proposition \ref{prop: hatYp-kn}, it can be shown that the following bound holds almost surely for any real numbers $s_{1, n}' \leq s_{2, n}'$,
\begin{align*}
\dK\Big(\L( M_{k_n}^*(\hat Z)|X),\L(M_{p}^*(\hat Z)|X)\Big) \  \leq  \ \P\big(A'(s_{2, n}') |X \big) \, + \, \P\big(B'(s_{1, n}') |X \big).
\end{align*}
If we choose 
\begin{align*}
s_{1,n}'=c_1' k_n^{-\beta(1-\tau)/4} \log(n)^{3/2}, \quad s_{2,n}'=c_2'l_n^{-\beta(1-\tau)} \sqrt{\log(d_n)},
\end{align*}
where $c_1', c_2' > 0$ are constants not depending on $n$, then $s_{1,n}' \leq s_{2,n}'$ holds when $n$ is sufficiently large (regardless of the particular values of $c_1'$ and $c_2'$). Recall also that $d_n$ is defined in~\eqref{eqn:dndef}.
Lemma \ref{Lemma: B'(t1') order} shows that there is a choice of $c_1'$ such that random variable $\P\big(B'(s_{1, n}') |X\big)$ is at most $n^{-1}$ with probability at least $1-c n^{-(1+\delta)}$. To deal with and $\P\big(A'(s_{2, n}') |X\big)$, notice that
\begin{align*}
\P\big(A'(s_{2, n}') |X\big) \leq  \P(\tilde M_{k_n}(X) \leq s_{2,n}') \ + \  \dK(\L(\tilde M_{k_n}(X),\L(M_{k_n}^*(\hat Z)|X)).
\end{align*}
Lemma \ref{Lemma: tilde MX} shows there is a choice of $c_2'$ such that the first term on the right hand side is of order $n^{-1/2}$. Finally, the second term is of order $n^{-1/2+\epsilon}$ with probability at least $1-cn^{-\delta/4}$ by Proposition \ref{prop: gaus-compa} and Lemma \ref{lemma: XZ kn-kn}.
\qed

\begin{lemma}
\label{Lemma: B'(t1') order}
If the conditions of Theorem~\ref{thm:main} hold, then there are constants $c,c_1'>0$ not depending on $n$ such that the following event holds with probability at least $1-cn^{-(1+\delta)}$ when $s_{1,n}'=c_1' k_n^{-\beta(1-\tau)/4} \log(n)^{3/2}$,
\begin{align*}
\P\big(B'(s_{1,n}') | X\big) \leq \frac{1}{n}.
\end{align*}
\end{lemma}

\begin{proof}
Notice that 
\begin{align*}
\max_{j\in J(k_n)^c}\frac{\sum_{i=1}^n \xi_i (\hat Z_{ij}-\hat  Z_j)}{\sqrt{n}\hat\sigma_j^{\tau}} \leq \max_{j\in J(k_n)^c} \frac{\sigma_j^{(\tau+1)/2}}{\hat \sigma_j^{\tau}} \ \cdot \max_{j\in J(k_n)^c}\frac{\sum_{i=1}^n \xi_i (\hat Z_{ij}-\hat  Z_j)}{\sqrt{n} \sigma_j^{(\tau+1)/2}}.
\end{align*}
The first factor on the right side is of order 1 with probability at least $1-cn^{-(2+\delta)}$ by Lemma \ref{Lemma: MOM var anti-con}. To handle the second factor, let $q=\max\{2, (1+\delta)\log(n)\}$. The idea of the rest of the proof is to construct a number $b_n$ such that the following event holds with high probability for every realization of the data,
\begin{align*}
\Big\|\max_{j\in J(k_n)^c}\frac{\sum_{i=1}^n \xi_i (\hat Z_{ij}-\hat  Z_j)}{\sqrt{n} \sigma_j^{(\tau+1)/2}} \Big\|_{L^q|X} \leq b_n.
\end{align*}
This will lead to the statement of the lemma via Chebyshev's inequality, because it will turn out that $s_{1,n}' \asymp b_n$. To construct $b_n$, first observe that
\begin{align*}
\Big\|\max_{j\in J(k_n)^c}\frac{\sum_{i=1}^n \xi_i (\hat Z_{ij}-\hat  Z_j)}{\sqrt{n} \sigma_j^{(\tau+1)/2}} \Big\|_{L^q|X}^q 
& \leq \sum_{j \in J(k_n)^c} \sigma_j^{-q (\tau+1)/2} \Big\|\frac{\sum_{i=1}^n \xi_i (\hat Z_{ij}-\hat  Z_j)}{\sqrt{n}}\Big\|_{L^q|X}^q  \\
& \leq c^q q^{3q/2}  \sum_{j \in J(k_n)^c} \sigma_{j}^{(1-\tau)q/4} \\
& \leq c^q \frac{q^{3q/2} }{(1-\tau)q\beta/4-1}  k_n^{1-(1-\tau)q\beta/4} 
\end{align*}
where the second inequality holds with probability at least $1-cn^{-(1+\delta)}$ by Lemma \ref{Lemma: Lq norm gaussian hat Z}, and the third inequality follows from the fact that $q\beta(1-\tau)/4 > 1$ when $n$ is sufficiently large. Since $q \asymp \log(n)$, we have
$$
\Big(\frac{1}{(1-\tau)q\beta/4-1} k_n \Big)^{1/q} \lesssim 1,
$$
and so the event
\begin{align*}
\bigg\|\max_{j\in J(k_n)^c}\frac{\sum_{i=1}^n \xi_i (\hat Z_{ij}-\hat  Z_j)}{\sqrt{n} \sigma_j^{(\tau+1)/2}} \bigg\|_{L^q|X} \leq c q^{3/2} k_n^{-(1-\tau)\beta/4}
\end{align*}
holds with probability at least $1-cn^{-(1+\delta)}$. Thus, we may take $b_n$ to be of the form $b_n = c q^{3/2} k_n^{-(1-\tau)\beta/4}$, and there is a choice of $c_1'$ such that the stated result holds.
\end{proof}

\begin{lemma}
\label{Lemma: Lq norm gaussian hat Z}
Let $q=\max\{2, (1+\delta)\log(n)\}$ and suppose the conditions of Theorem~\ref{thm:main} hold. Then, there is a constant $c>0$ not depending on $n$ such that the bound 
$$
\bigg\|\frac{1}{\sqrt n} \sum_{i=1}^n \xi_i (\hat Z_{ij}-\hat  Z_j)\bigg\|_{L^q|X} 
 \leq cq^{3/2} \sigma_j^{\frac{\tau+3}{4}}
$$
holds simultaneously over all $j\in J(k_n)^c$ with probability at least $1-cn^{-(1+\delta)}$.
\end{lemma}

\noindent \emph{Proof.}
The $L^q$-norm bound for centered Gaussian random variables in~\eqref{eqn:gaussianLq} gives
\begin{align}\label{eqn:firststepZhat}
\bigg\|\frac{\sum_{i=1}^n \xi_i (\hat Z_{ij}-\hat  Z_j)}{\sqrt{n}}\bigg\|_{L^q|X} 
\ \leq \ c \sqrt{q} \bigg( \frac{1}{n} \sum_{i=1}^n \hat Z_{ij}^2 \bigg)^{1/2}.
\end{align}
To develop a high-probability bound for the right side, note that we have $|\hat Z_{ij}| \leq |\hat Y_{ij}| + |\tilde X_j |$ for any fixed $j$, and so
\begin{align*}
\frac{1}{n} \sum_{i=1}^n \hat Z_{ij}^2 \ \leq \  \frac{2}{n}\sum_{i=1}^n \hat Y_{ij}^2 \ + \ 2|\tilde X_j |^2.
\end{align*}
Here, we apply Lemma \ref{Lemma: MOM mean}\eqref{Equa: MOM mean 2} with $\theta=(1-\tau)/4$ in the notation used there. This implies there is a constant $c>0$ not depending on $n$ such that the bound
\begin{equation*}
|\tilde X_j|^2 \leq c \sigma_j^{\frac{\tau+3}{2}}
\end{equation*}
holds simultaneously over all $j\in J(k_n)^c$ with probability at least $1-cn^{-(1+\delta)}$. Next, Rosenthal's inequality for non-negative random variables (Lemma~\ref{Lemma: Rosenthal's inequality})  gives
\begin{equation*}
\begin{split}
    \bigg\|\frac{1}{n}\sum_{i=1}^n \hat Y_{ij}^2\bigg\|_{L^q|X'} 
    & \ \leq \  c q\max\Big\{ \|\hat Y_{1j}^2\|_{L^1|X'}\, , \, n^{-1+1/q}\|\hat Y_{1j}^2\|_{L^q|X'}\Big\}\\[0.1cm]
& \ \leq \ cq\max\Big\{ \|X_{1j}^2\|_{L^1}\, , \, n^{-1+1/q}\hat t_j^2\Big\}\\[0.1cm]
& \ \leq \ cq (\sigma_j^2+\hat\sigma_j^2),
    \end{split}
\end{equation*}
where the second step uses $|\hat Y_{1j}|\leq |X_{1j}|\wedge \hat t_j$. Applying Chebyshev's inequality conditionally on $X'$ shows that the event
\begin{equation*}
    \P\Big(\frac{1}{n}\sum_{i=1}^n \hat Y_{ij}^2 \ \geq \ ceq(\sigma_j^2+\hat\sigma_j^2) \sigma_{j}^{\frac{\tau-1}{4}} \Big|X'\Big) \ \leq \ e^{-q} \sigma_{j}^{\frac{(1-\tau)q}{4}} \ \leq  cn^{-(1+\delta)} \sigma_{j}^{\frac{(1-\tau)q}{4}}
    \end{equation*}
holds with probability 1, and thus the unconditional version of the left hand side is also at most $cn^{-(1+\delta)} \sigma_{j}^{\frac{(1-\tau)q}{4}}$. Consequently, the event
\begin{align*}
\frac{1}{n}\sum_{i=1}^n \hat Y_{ij}^2 \ \leq \ ceq(\sigma_j^2+\hat\sigma_j^2) \sigma_{j}^{\frac{\tau-1}{4}} 
\end{align*}
holds simultaneously over all $j\in J(k_n)^c$ with probability at least $1-cn^{-(1+\delta)}$ since 
$$
\sum_{j \in J(k_n)^c} j^{-\frac{(1-\tau)q \beta}{4}} \lesssim k_n^{-\frac{(1-\tau)q\beta}{4}+1} \lesssim 1.
$$
Finally, Lemma~\ref{Lemma: MOM variance}\eqref{Equa: MOM variance 2} implies there is a constant $c>0$ not depending on $n$ that the bound
$$\hat\sigma_j^2\leq c\sigma_j^{\frac{\tau+7}{4}}$$
holds simultaneously over all $j\in J(k_n)^c$ with probability at least $1-cn^{-(2+\delta)}$. Combining results above, we have that the bound
\begin{align}
\label{equa: sum hatZij}
\frac{1}{n} \sum_{i=1}^n \hat Z_{ij}^2 \ \leq \  cq \sigma_j^{\frac{\tau+3}{2}}
\end{align}
holds simultaneously over $j\in J(k_n)^c$ with probability at least $1-cn^{-(1+\delta)}$, which completes the proof.\qed

\begin{lemma}
\label{Lemma: moments closeness} If the conditions of Theorem~\ref{thm:main} hold, then there is a constant $c>0$ not depending on $n$ such that the following bounds hold for all $j, k \in \{1,\dots,p\}$, 
\begin{align*}
|\E(Y_{1j})| & \, \leq \,  c \sigma_j n^{-\frac{3}{2}} \\
|\E(\hat Y_{1j} | X')| & \, \leq \, c\sigma_j^4 \hat \sigma_j^{-3} n^{-\frac{3}{2}} \ \ \textup{ (almost surely)}\\
|\Cov(Y_{1j}, Y_{1k}) - \Cov(X_{1j}, X_{1k})| & \, \leq \,  c \sigma_j \sigma_k n^{-1}.
\end{align*}
\end{lemma}

\noindent \emph{Proof.} We may assume without loss of generality that $\E(X_{1j})=0$ for all $j\in\{1,\dots,p\}$. Observe that Assumption~\ref{A:model}(\ref{A:moments}) gives
\begin{align*}
|\E(Y_{1j})|
& \leq \E (|X_{1j}| 1 \{ |X_{1j}| \geq t_j\})  \\
& \leq \| X_{1j}\|_{L^4}\| 1\{ |X_{1j}| \geq t_j\}\|_{L^{4/3}}\\
& \leq c \sigma_j \Big(\ts\frac{\|X_{1j}\|_{L^4}^4}{t_j^4}\Big)^{3/4}\\
& \leq c\sigma_j n^{-\frac{3}{2}}.
\end{align*}
Second, the stated bound on $|\E(\hat Y_{1j} | X')|$ can also be obtained from essentially the same argument.
Third, to bound the difference between the covariances, we use the fact that $|Y_{1j}Y_{1k}-X_{1j}X_{1k}|$ vanishes on the intersection of the events $\{|X_{1j}|\leq t_j\}$ and $\{|X_{1k}|\leq t_{k}\}$, and otherwise it is at most $2|X_{1j}X_{1k}|$. Therefore, we have
\begin{align*}
\big|\E(Y_{1j}Y_{1k})-\E(X_{1j}X_{1k})\big| & 
\ \leq \  2\E\big(|X_{1j} X_{1k}| 1 \{|X_{1j}| \geq t_j\}\big) + 2\E\big(|X_{1j} X_{1k}| 1 \{|X_{1k}| \geq t_k\}\big).
\end{align*}
The two terms on the right hand side can be handled via
\begin{align*}
\E\big(|X_{1j} X_{1k}| 1 \{|X_{1j}| \geq t_j\}\big) & \leq \|X_{1j} \|_{L^4} \|X_{1k}\|_{L^4} \| 1 \{|X_{1j}| \geq t_j\} \big\|_{L^2} \\[0.2cm]
&\leq c\sigma_j\sigma_k\Big(\ts\frac{\|X_{1j}\|_{L^4}^4}{t_j^4}\Big)^{1/2}\\
&\leq \frac{c \sigma_j \sigma_k}{n},
\end{align*}
which yields the stated result.
\qed

\begin{lemma}
\label{Lemma: Lq of Snj}
If the conditions of Theorem~\ref{thm:main} hold and $q=\max\{\log(n),2\}$, then
$$
\bigg\| \sum_{i=1}^n \frac{ Y_{ij} - \E(Y_{ij})}{\sqrt{n} \sigma_{j}} \bigg\|_{L^q} \lesssim q.
$$
\end{lemma}

\begin{proof}
Due to Rosenthal's inequality (Lemma \ref{Lemma: Rosenthal's inequality}), we have 
\begin{align*}
\bigg\| \sum_{i=1}^n \frac{ Y_{ij} - \E(Y_{ij})}{\sqrt{n} \sigma_{j}} \bigg\|_{L^q} 
& \lesssim q \max \Bigg\{\bigg\|  \sum_{i=1}^n \frac{ Y_{ij} - \E(Y_{ij})}{\sqrt{n} \sigma_{j}} \bigg\|_{L^2}, \bigg(\sum_{i=1}^n  \bigg\| \frac{Y_{ij}-\E(Y_{ij})}{ \sqrt{n} \sigma_{j}} \bigg\|_{L^q}^q\bigg)^{1/q} \Bigg\} \\
& \lesssim q \max \big\{1 , n^{1/q} \big\} \\
& \lesssim q,
\end{align*}
where the second step uses the almost-sure bound $|Y_{ij}|\leq \sqrt n \sigma_j$, as well as Lemma~\ref{Lemma: moments closeness} to relate the variance of $Y_{ij}$ with $\sigma_j$.
\end{proof}

\section{Results on median-of-means estimators}\label{sec:mom}
The results in this section will continue to follow the convention that $\max_{1\leq j\leq p}\sigma_j^2=1$, as discussed on p.\pageref{sigmaconvention}. However, to make the results easier to interpret, we will state them so that they explicitly account for the coordinate-wise means $\mu_j=\E(X_{1j})$, $j=1,\dots,p$ (even though these parameters may be assumed to be zero without loss of generality in proving Theorem~\ref{thm:main}).

\begin{lemma}
\label{Lemma: MOM mean}
Fix any constant $\theta\in (0,1)$ and suppose that the conditions of Theorem~\ref{thm:main} hold. Then, there is a constant $c\geq 1$ not depending on $n$, such that for any $j\in \{1,\dots,p\}$, the median-of-means estimator $\tilde X_{j}$  with $b_n\asymp \log(n)$ blocks satisfies 
\begin{align}
\P\Big(|\tilde X_{j}-\mu_j| \geq   \sigma_{j} n^{-1/2+\epsilon/2}\Big) 
& \ \lesssim \ \Big(\frac{c}{n}\Big)^{\frac{\epsilon}{c}b_n} \tag{i}\label{Equa: MOM mean 1}\\[0.2cm]
 \P\Big(|\tilde X_{j}-\mu_j| \geq    \sigma_j^{1-\theta} n^{-1/2+\epsilon/2}  \Big)
 & \ \lesssim \ (c\sigma_j)^{\frac{\theta}{c}b_n} \tag{ii}\label{Equa: MOM mean sigma}.
\end{align}
Furthermore, we have
\begin{align}\sum_{j \in J(k_n)}\P\Big(|\tilde X_{j}-\mu_j| \geq   \sigma_{j} n^{-1/2+\epsilon/2}\Big) 
& \ \lesssim \ n^{-(2+\delta)}
\tag{iii}\label{Equa: MOM mean 1 sum} \\[0.2cm]
 \sum_{j=1}^p \P\Big(|\tilde X_{j}-\mu_j| \geq C   \sigma_j^{1-\theta} n^{-1/2+\epsilon/2}\Big)
 & \ \lesssim \  k_n^{-\log(n)/c}. \tag{iv}\label{Equa: MOM mean 2}
\end{align}
\end{lemma}

\noindent \emph{Proof.} Recall the notation $\bar X_j(l)=\frac{1}{\ell_n}\sum_{i\in \mathcal{B}_l} X_{ij}$ where $l=1,\dots,b_n$.  Fix $t>0$ and let $\xi_{jl}=1\big\{ |\bar X_j(l)-\mu_j| \geq  t \big\}$. Since the event $\{|\tilde X_j(l)-\mu_j|\geq t\}$ can only occur if at least half of the random variables $\xi_{j1},\dots,\xi_{jb_n}$ are 1, we must have
\begin{align*}
    \P\big( |\tilde X_{j}-\mu_j| \geq t \big) 
& \ \leq \ \P\bigg(\frac{1}{b_n} \sum_{l=1}^{b_n} \xi_{jl} \ \geq \frac{1}{2}\bigg).
\end{align*}
Next, applying Kiefer's inequality (Lemma~\ref{Lemma: Kiefer's inequality}) to the right side gives
\begin{equation}\label{eqn:Kiefer}
    \P( |\tilde X_{j}-\mu_j| \geq t ) 
     \ \lesssim \ \big(e\E(\xi_{1j})\big)^{b_n(\frac{1}{2}-\E(\xi_{j1}))^2}.
\end{equation}
Furthermore, by Chebyshev's inequality, $\E(\xi_{jl})\lesssim \  \frac{\sigma_j^2}{\ell_n t^2}$, and so if we take $t=\sigma_j n^{-1/2+\epsilon/2}$, then
\begin{equation*}
\begin{split}
    \E(\xi_{jl})
    \ \lesssim \ \frac{ n^{1-\epsilon}}{\ell_n} \ \lesssim \  n^{-\epsilon/2},
    \end{split}
\end{equation*}
where the last step uses $n/\ell_n\asymp m_n/\ell_n=b_n\asymp \log(n)$. Thus, combining this bound on $\E(\xi_{jl})$ with~\eqref{eqn:Kiefer} establishes the first claim~\eqref{Equa: MOM mean 1}. Similarly, choosing $t=\sigma_j^{1-\theta}n^{-1/2+\epsilon/2}$ in the previous argument leads to the second claim~\eqref{Equa: MOM mean sigma}.

For the fourth claim~\eqref{Equa: MOM mean 2}, we decompose the sum over $j=1,\dots,p$ along the indices in $J(k_n)$ and $J(k_n)^c$. To bound the sum over $J(k_n)^c$, we may use \eqref{Equa: MOM mean sigma} to obtain
\begin{align*}
\sum_{j\in J(k_n)^c}\P\Big( |\tilde X_{j}-\mu_j| \geq C \sigma_{j}^{1-\theta} n^{-1/2+\epsilon/2} \Big) 
& \ \lesssim \ \sum_{j\in J(k_n)^c} (c\sigma_j)^{\frac{\theta}{c}b_n}\\[0.2cm]
& \ \lesssim \  \sum_{j\geq k_n} (Cj^{-\beta})^{\frac{\theta}{c}b_n}\\[0.2cm]
& \ \lesssim \ k_n^{{-\frac{\beta\theta}{c}b_n}+1}\\[0.2cm]
& \ \lesssim \ k_n^{-\log(n)/c}.
\end{align*}

 Regarding the sum over $j\in J(k_n)$, note that $C\sigma_j^{1-\theta}\geq \sigma_j$ holds for all $j=1,\dots,p$. Therefore the bound~\eqref{Equa: MOM mean 1} gives
\begin{align*}
\sum_{j\in J(k_n)}\P\Big( |\tilde X_{j}-\mu_j| \geq  C\sigma_{j}^{1-\theta} n^{-1/2+\epsilon/2} \Big) 
& \ \leq \sum_{j\in J(k_n)}\P\Big( |\tilde X_{j}-\mu_j| \geq  C\sigma_{j} n^{-1/2+\epsilon/2} \Big) \\
& \ \lesssim \  k_n \Big(\frac{c}{n}\Big)^{\frac{\epsilon}{c}b_n}\\
& \ \lesssim \ n^{-(2+\delta)}.
\end{align*}
This leads to the third claim \eqref{Equa: MOM mean 1 sum} and completes the proof. \qed
\begin{lemma}
\label{Lemma: MOM variance}
Fix any constant $\theta\in (0,1)$, and suppose that the conditions of Theorem~\ref{thm:main} hold. Then, there is a constant $c\geq 1$ not depending on $n$, such that the following bounds hold for any $j\in\{1,\dots,p\}$,
\begin{align}
 \P\Big(|\hat{\sigma}_j^2-\sigma_j^2 |>  \sigma_j^2 n^{-1/2+\epsilon/2} \Big) & \ \lesssim \ \Big(\frac{c}{n}\Big)^{\frac{\epsilon}{c}b_n} \tag{i}\label{Equa: MOM variance 1}\\[0.2cm]
 \P\Big(|\hat{\sigma}_j^2-\sigma_j^2 | \geq   \sigma_{j}^{2-2\theta} n^{-1/2+\epsilon/2}\Big) & \ \lesssim \ (c\sigma_j^2)^{\frac{\theta}{c}b_n}.\tag{ii}
      \end{align}    
 Furthermore, we have 
\begin{align}
\sum_{j \in J(k_n)} \P\Big(|\hat{\sigma}_j^2-\sigma_j^2 |> C^2\sigma_j^{2} n^{-1/2+\epsilon/2} \Big) 
& \ \lesssim \ n^{-(2+\delta)}.\tag{iii}\label{Equa: MOM variance 1 sum}\\[0.2cm]
 \sum_{j=1}^p \P\Big(|\hat{\sigma}_j^2-\sigma_j^2 |> C^2\sigma_j^{2-2\theta} n^{-1/2+\epsilon/2} \Big) 
& \ \lesssim \ k_n^{-\log(n)/c}.\tag{iv}\label{Equa: MOM variance 2}
\end{align}
\end{lemma}
\noindent \emph{Proof.} The proof of Lemma~\ref{Lemma: MOM mean} can be repeated with the i.i.d.~random variables $\frac{1}{2}(X_{ij} - X_{i'j})^2$, playing the role that $X_{ij}$ previously did. Also note that because  $\var\big(\frac{1}{2}(X_{ij}-X_{i'j})^2\big)\lesssim \sigma_j^4$ holds under Assumption~\ref{A:model}(\ref{A:moments}), the parameter $\sigma_j^4$ plays the role that $\sigma_j^2$ did in the context of Lemma~\ref{Lemma: MOM mean}.
\qed

For the next lemma, recall that for each $l\in\{1,\dots,b_n\}$, the $l$th blockwise variance estimate $\bar\sigma_j^2(l)$ for $\sigma_j^2$ is defined in~\eqref{eqn:barsigmadef}.
\begin{lemma}
\label{Lemma: var lower bound}
Fix any constant $\theta\in (0,1)$, and suppose that the conditions of Theorem~\ref{thm:main} hold. Then, there is a constant $c\geq 1$ not depending on $n$, such that the following bound holds for any $j\in\{1,\dots,p\}$ and any $l\in\{1,\dots,b_n\}$,    $$\P\big(\bar\sigma_{j}^2(l)\leq \sigma_j^{2+2\theta}\big)
\ \lesssim \  (c \ell_n\sigma_j^{2\theta})^{\ell_n/4}.
$$
\end{lemma}

\noindent\emph{Proof.} Because the $\ell_n/2$ terms in the definition of $\bar\sigma_j^2(l)$ are i.i.d., it follows that
\begin{equation}\label{eqn:summax}
\begin{split}
\P\big(\bar\sigma_{j}^2(l\big ) \ \leq \  \sigma_j^{2+2\theta}) & \ \leq \
    \P\Bigg(\max_{\substack{i,i'\in \mathcal{B}_l\\ i'-i=\ell_n/2}} \textstyle \frac{1}{\ell_n}(X_{ij}-X_{i'j})^2\leq \sigma_j^{2+2\theta} \Bigg)\\[0.2cm]
    & \ = \   \P\bigg( \textstyle\frac{1}{2\sigma_j^2}(X_{1j}-X_{(\ell_n/2+1)j})^2\leq \frac{1}{2}\ell_n \sigma_j^{2\theta}\bigg)^{\ell_n/2}.
    \end{split}
\end{equation}
Since the independent random variables $X_{1j}/\sigma_j$ and $X_{(1+\ell_n/2)j}/\sigma_j$ have densities whose $L^{\infty}$ norms are $\mathcal{O}(1)$ under Assumption~\ref{A:model}(\ref{A:conti}), it follows from Young's convolution inequality~\citep[p.178]{stein1971introduction} that the random variable $\frac{1}{\sqrt 2\sigma_j}(X_{1j}-X_{(\ell_n/2+1)j})$ also has a density whose $L^{\infty}$  norm is $\mathcal{O}(1)$, and so
\begin{align*}
    \P\bigg( \frac{1}{2\sigma_j^2}\big(X_{1j}-X_{(\ell_n/2+1)j}\big)^2\leq \frac{1}{2} \ell_n \sigma_j^{2\theta}\bigg) 
    & \, \lesssim \, \big(\textstyle \ell_n\sigma_j^{2\theta}\big)^{1/2}.
\end{align*}
Combining this with~\eqref{eqn:summax} completes the proof.\qed

\begin{lemma}
\label{Lemma: MOM var anti-con}
Fix any constant $\theta\in (0,1)$, and suppose that the conditions of Theorem~\ref{thm:main} hold. Then, there is a constant $c\geq 1$ not depending on $n$, such that the event
$$\max_{1\leq j\leq p } \frac{\sigma_j^{2+2\theta}}{\hat \sigma_j^2} \leq c$$
holds with probability at least $1-cn^{-(2+\delta)}$.
\end{lemma}
\noindent \emph{Proof.} Let $r_n=\lceil n^{\epsilon/(\theta \beta)} \wedge p\rceil$. It follows from Lemma~\ref{Lemma: MOM variance}\eqref{Equa: MOM variance 1} that there is a constant $c>0$ not depending on $n$ such that the bound
\begin{equation}\label{eqn:sigratiobound}
    \max_{j\in J(r_n)}\frac{\sigma_j^2}{\hat\sigma_j^2}\leq c
\end{equation}
holds with probability at least $1-cn^{-(2+\delta)}$. Therefore, a bound of the same form must also hold for  $\max_{j\in J(r_n)}\frac{\sigma_j^{2+2\theta}}{\hat\sigma_j^2}$, since $\max_{1\leq j\leq p}\sigma_j^{2\theta}\lesssim 1$.

To complete the proof, it remains to handle the maximum of $\sigma_j^{2+2\theta}/\hat\sigma_j^2$ over indices $j$ in the complementary set $J(r_n)^c$.
Letting $\xi_{jl}=1\{\bar\sigma_{j}^2(l)\leq \sigma_j^{2+2\theta}\}$ for $l=1,\dots,b_n$,
 Kiefer's inequality (Lemma \ref{Lemma: Kiefer's inequality}) implies that the following bound holds for any $j\in J(r_n)$,
\begin{align}\label{eqn:sigmajlower}
    \P\big(\hat{\sigma}_j^2\leq \sigma_j^{2+2\theta}\big)  
    \ \leq \ \P\bigg(\frac{1}{b_n}\sum_{l=1}^{b_n} \xi_{jl} \geq \frac{1}{2}\bigg) \ \lesssim \ (e\E(\xi_{j1}))^{b_n(\frac{1}{2}-\E(\xi_{j1}))^2}.
\end{align}
Also, Lemma \ref{Lemma: var lower bound} gives 
\begin{align*}
    \E(\xi_{j1})
    & \ \lesssim \ (c \ell_n j^{-2\theta\beta})^{\ell_n/4}.
\end{align*}
Therefore, combining with with~\eqref{eqn:sigmajlower}, we conclude that
\begin{align*}
\sum_{j\in J(r_n)^c}  \P(\hat{\sigma}_j^2\leq \sigma_j^{2+\theta} )
 & \ \lesssim \  (c\ell_n)^{\ell_n/4}\sum_{j\geq  n^{\frac{\epsilon}{\theta \beta}} } (j^{-\theta\beta\ell_n/2})^{b_n/c}\\
 & \ \lesssim \ (c\ell_n)^{\ell_n/4} n^{-\frac{n}{2c}+1} \\
 & \ \lesssim n^{-(2+\delta)},
\end{align*}
where the last step uses $\ell_n\asymp n/\log(n)$. Note also that the final bound $n^{-(2+\delta)}$ can be replaced with any fixed positive power of $n^{-1}$, but the current form is all that is needed.
\qed

\section{Proof of Proposition 1}
\label{supp: moments eg}
\noindent To ease notation, we let $q=4+\delta$ throughout the proof.\\

\noindent\textbf{Elliptical case.} Suppose $X_1$ is a centered elliptical random vector of the form $X_1=\eta_1\Sigma^{1/2}Z_1/\|Z_1\|_2$, where $Z_1$ is a standard Gaussian $p$-dimensional Gaussian vector, and $\eta_1$ is independent of $Z_1$ with $\E(\eta_1^2)=p$. 
 We first check the $L^{q}$-$L^2$ moment equivalence condition~(i) with $q=4+\delta$.  Letting $w=\Sigma^{1/2}v$ for a generic vector $v\in\R^p$, a direct calculation gives
\begin{align*}
\|\langle v, X_1 \rangle\|_{L^2}^2 
&=\|\eta_1 \langle w, Z_1/\|Z_1\|_2\rangle\|_{L^2}\\
&= \|\eta_1\|_{L^2}^2 \| w \|_2^2/p\\
&= \| w \|_2^2.
\end{align*}
Because the distribution of $U_1=Z_1/\|Z_1\|_2$ is invariant to orthogonal transformations, it follows that the random variables $\langle w, U_1\rangle$ and $ \|w\|_2\langle e_1, U_1\rangle$ are equal in distribution, where $e_1$ is the first standard basis vector. Therefore,
\begin{align*}
\|\langle v, X_1 \rangle \|_{L^{q}} 
& = \|w\|_2 \|\eta_{1}\|_{L^{q}} \|U_{11}\|_{L^{q}}.
\end{align*}
The quantity $\|U_{11}\|_{L^q}$ is at most of order $1/\sqrt p$, which can be shown as follows. Due to the independence of $U_{11}$ and $\|Z_1\|_2$, we have
$\|Z_{11}\|_{L^{q} } = \| \| Z_1\|_2\|_{L^{q} }  \|U_{11}\|_{L^{q} } $. Furthermore, Lyapunov's inequality gives $\|\|Z_1\|_2\|_{L^q}= \| \| Z_1\|_2^2 \|_{L^{q/2}}^{1/2} \geq  \| \| Z_1\|_2^2 \|_{L^{1}}^{1/2} = \sqrt{p}$. 
Combining the last several steps and the assumption that $\|\eta_1\|_{L^q}\lesssim \sqrt p$, we conclude that
\begin{equation}
    \|\langle v, X_1 \rangle \|_{L^{q}} \lesssim \|w\|_2=\|\langle v, X_1 \rangle\|_{L^2},
\end{equation}
which verifies condition~(i). 
Regarding the density condition (ii), note that if $e_j$ is the $j$th standard basis vector, then the vector $w=\Sigma^{1/2}e_j/\sigma_j$ satisfies $\|w\|_2=1$. So, the discussion above shows that all the random variables  $X_{11}/\sigma_1,\dots,X_{1p}/\sigma_p$ have the same distribution, which is that of $\eta_1 \langle e_1,U_1\rangle$. In particular, if the random variable $X_{11}/\sigma_1$ has a Lebesgue density $f_1$ such that $\|f_1\|_{L^{\infty}}\lesssim 1$, then condition (ii) holds, which completes the proof in the elliptical case. 

\noindent\textbf{Separable case.} Suppose that $X_1$ has a centered separable distribution so that \smash{$X_1=\Sigma^{1/2}\zeta_1$,} where $\zeta_1=(\zeta_{11},\dots,\zeta_{1p})$ has i.i.d.~entries with $\E(\zeta_{11})=0$, and $\Var(\zeta_{11})=1$.
To check the $L^{q}$-$L^2$ moment equivalence condition~(i), it suffices to show that $\|\langle w,\zeta_1\rangle \|_{L^{q}} \lesssim   \|\langle w, \zeta_1 \rangle \|_{L^{2}}$ for any $w\in\R^p$.  Using Rosenthal's inequality (Lemma \ref{Lemma: Rosenthal's inequality}), and the assumption that $\max_{1 \leq j \leq p} \|\zeta_{1j}\|_{L^{q}} \lesssim 1$, we have
        \begin{align*}
        \|\langle w, \zeta_1 \big\rangle \|_{L^{q}} & \lesssim \max \bigg\{ \|\langle w, \zeta_1 \rangle \big\|_{L^{2}}, \Big( \sum_{j=1}^p \big\| w_j \zeta_{1j} \|_{L^{q}}^{q}  \Big)^{1/q}  \bigg\} \\[0.1cm]
        & \lesssim \max \big\{ \|w\|_2, \|w \|_{q} \big\} \\[0.1cm]
        & =  \|\langle w, \zeta_1 \rangle \|_{L^{2}}
        \end{align*}
        as needed. Finally, the density condition~(ii) is a direct consequence of Theorem~1.2 in the paper~\citep{rudelson2015small} and the assumption that $\max_{1\leq j\leq p}\|g_j\|_{L^{\infty}}\lesssim 1$, where $g_j$ is the Lebesgue density of $\zeta_{1j}$. \qed
\section{Background results}\label{sec:background}
\begin{lemma}[Rosenthal inequalities~\citep{Rosenthal}]
\label{Lemma: Rosenthal's inequality}
Fix $q \geq 1$, and let $\xi_1, \ldots, \xi_n$ be independent random variables. Then, there is an absolute constant $c>0$ such that the following two statements are true.

(i). If $\xi_1, \ldots, \xi_n$ are non-negative,
then $$\Big\|\sum_{i=1}^n \xi_i \Big\|_{L^q}
\leq c \cdot q \cdot \max \bigg\{\Big\|\sum_{i=1}^n \xi_i\Big\|_{L^1},\Big(\sum_{i=1}^n \big\|\xi_i \big\|_{L^q}^q\Big)^{1 / q} \bigg\}.$$

(ii). If $q\geq 2$, and $\xi_1, \ldots, \xi_n$ are centered, then
$$
\Big\|\sum_{i=1}^n \xi_i \Big\|_{L^q} 
\leq c\cdot q \cdot \max \bigg\{\Big\|\sum_{i=1}^n \xi_i\Big\|_{L^2},\Big(\sum_{i=1}^n \big\|\xi_i \big\|_{L^q}^q\Big)^{1 / q}\bigg\}.
$$
\end{lemma}

\noindent For the next lemma, recall that we denote the stable rank of a non-zero positive semidefinite matrix $A$ as ${\tt{r}}(A)=\textup{tr}(A)^2/\|A\|_F^2$.
\begin{lemma}[Lower-tail bound for Gaussian maxima~\citep{lopes2022sharp}]
\label{lem:lowertail}
Let $\xi \sim \mathcal{N}(0, R)$ be a Gaussian random vector in $\R^d$ for some correlation matrix $R$, and fix two constants $a, b \in(0,1)$ with respect to $d$. Then, there is a constant $c>0$ depending only on $(a, b)$ such that the following inequality holds for any integer $k$ satisfying
$2\leq k \leq \frac{b^2}{4} {\tt{r}}(R)$,
$$
\P\Big(\max_{1\leq j\leq d} \xi_j \leq a \sqrt{2(1-b) \log (k)}\Big)
\leq c\, k^{\frac{-(1-b)(1-a)^2}{b}}\big(\log (k)\big)^{\frac{1-b(2-a)-a}{2 b}}
 $$
\end{lemma}

\noindent The next result is a variant of Lemma A.7 in \cite{spokoiny2015bootstrap} that can be proven in essentially the same way.

\begin{lemma}[Gaussian comparison inequality]
\label{lem:GaussianFrobenius}
Let $\zeta\sim N(0,I_d)$ and $\xi\sim N(0,A)$ for some positive semidefinite matrix $A\in \R^{d\times d}$. Then, 
$$
\sup _{s \in \mathbb{R}}\bigg|\P\Big(\max _{1 \leq j \leq d} \xi_j \leq s\Big)-\P\Big(\max _{1 \leq j \leq d} \zeta_j \leq s\Big)\bigg| \ \leq \ 2  \|A-I_d\|_F.
$$
\end{lemma}

\begin{lemma}[Bentkus' Berry-Esseen Theorem \citep{Bentkus:2003}]
\label{Lemma: Bentkus berry esseen}
Let $V_1, \ldots, V_n$ be i.i.d.~random vectors in $\mathbb{R}^d$ with zero mean and identity covariance matrix. Furthermore, let $\gamma_d$ denote the standard Gaussian distribution on $\mathbb{R}^d$, and let $\mathscr{A}$ denote the collection of all Borel convex subsets of $\mathbb{R}^d$. Then, there is an absolute constant $c>0$ such that
$$
\sup _{\mathcal{A} \in \mathscr{A}}\left|\P\Big(\frac{1}{\sqrt{n}}\sum_{i=1}^n V_i \in \mathcal{A}\Big)-\gamma_d(\mathcal{A})\right| \leq \frac{c \cdot d^{1 / 4} \cdot \E\left\|V_1\right\|_2^3}{n^{1 / 2}} .
$$    
\end{lemma}

\begin{lemma}[Nazarov's inequality]
\label{Lemma: Nazarov's inequality}
Let $(\zeta_1, \ldots, \zeta_d)$ be a Gaussian random vector, and suppose the parameter $\underline{\sigma}^2:=\min _{1 \leq j \leq d} \Var (\zeta_j)$ is positive. Then, for any fixed $\epsilon>0$,
\begin{align*}
\sup _{s \in \mathbb{R}} \P\Big(\big|\max _{1 \leq j \leq d} \zeta_j-s\big| \leq \epsilon \big) 
\ \leq \ \frac{2 \epsilon}{\underline{\sigma}} (\sqrt{2 \log (d)}+2) .
\end{align*}
\end{lemma}

\noindent This version of Nazarov's inequality appears in Lemma 4.3 of \citep[][]{chernozhukov2016empirical} and originates from~\cite{nazarov2003maximal}.

\begin{lemma}[Fuk-Nagaev inequality]
\label{Lemma: Fuk-Nagaev inequalities}
Fix $q \geq 1$, and let $\xi_1, \ldots, \xi_n$ be centered independent random variables. Then, for any fixed $s>0$,
\begin{align*}
\P\bigg(\Big|\sum_{i=1}^n \xi_i\Big| \geq s\bigg) \ \leq \ 2\Big( \frac{q+2}{qs} \Big)^q \sum_{i=1}^n \E|\xi_i|^q
+ 2\exp\bigg(\frac{-2s^2}{(q+2)^2 e^q \sum_{i=1}^n \E(\xi_{i}^2)}\bigg).
\end{align*}
\end{lemma}

\noindent This statement of the Fuk-Nagaev inequality is based on~\cite[][eqn.~1.7]{rio2017constants}.

\begin{lemma}[Kiefer's inequality]
\label{Lemma: Kiefer's inequality}
If $\xi_1, \ldots, \xi_n$ are i.i.d.~Bernoulli random variables, then 
$$
\P\bigg(\frac{1}{n}\sum_{i=1}^n \xi_i \geq \frac{1}{2}\Big) \ \leq \ 
2 \big(e\E(\xi_1)\big)^{n(1/2-\E(\xi_1))^2}.
$$
\end{lemma}
\noindent The result above is the modification of Corollary~A.6.3~in {\citep[][]{van1996weak}}. In that reference, the success probability $\E(\xi_1)$ of the Bernoulli random variables is assumed to be less than $1/e$, but in the formulation above, the result holds for all values of $\E(\xi_1)$.

\end{document}